\documentclass[journal]{IEEEtran}
\usepackage{cite}
\usepackage{hyperref}
\usepackage{url}
\usepackage{lineno}
\usepackage{amsmath,amssymb,amsfonts,amsthm}
\usepackage{mathtools}
\usepackage{cases}
\usepackage{bm}
\usepackage{mathrsfs}
\usepackage{dsfont}
\usepackage{graphicx}
\usepackage[caption=false,font=footnotesize]{subfig}
\usepackage{float}
\usepackage{array}
\usepackage{multirow}
\usepackage{algorithm}
\usepackage{algorithmic}
\usepackage{xcolor}
\usepackage{textcomp}
\usepackage{setspace}
\usepackage{framed}
\usepackage{balance}

\allowdisplaybreaks

\newtheorem{remark}{Remark}
\newtheorem{theorem}{Theorem}

\newtheorem{lemma}{Lemma}

\newtheorem{corollary}{Corollary}

\makeatletter
\newcommand{\biggg}{\bBigg@{3}}
\newcommand{\Biggg}{\bBigg@{3.5}}
\makeatother
\makeatletter
\renewcommand{\maketag@@@}[1]{\hbox{\m@th\normalsize\normalfont#1}}
\makeatother
\def\BibTeX{{\rm B\kern-.05em{\sc i\kern-.025em b}\kern-.08em
T\kern-.1667em\lower.7ex\hbox{E}\kern-.125emX}}
\newcounter{problem}
\newcounter{save@equation}
\newcounter{save@problem}

\makeatletter
\newenvironment{problem}
{\setcounter{problem}{\value{save@problem}}%
\setcounter{save@equation}{\value{equation}}%
\let\c@equation\c@problem
\subequations
}
{\endsubequations
\setcounter{save@problem}{\value{equation}}%
\setcounter{equation}{\value{save@equation}}%
}

\usepackage{acro}

\DeclareAcronym{mimo}{
    short = MIMO,
    long = multiple-input and multiple-output,
}
\DeclareAcronym{jcc}{
      short = JCC,
      long  = joint communication and control,
}
\DeclareAcronym{bs}{
    short=BS,
    long=base station
}
\DeclareAcronym{cu}{
    short=CU,
    long=communication user
}
\DeclareAcronym{sinr}{
    short=SINR,
    long=signal-to-interference-plus-noise ratio
}
\DeclareAcronym{snr}{
    short=SNR,
    long=signal-to-noise ratio
}
\DeclareAcronym{lmmse}{
    short=LMMSE,
    long=linear minimum mean-square error
}
\DeclareAcronym{svd}{
    short=SVD,
    long=singular value decomposition
}
\DeclareAcronym{kf}{
    short=KF,
    long=Kalman filter
}
\DeclareAcronym{csi}{
    short=CSI,
    long=channel state information
}
\DeclareAcronym{dare}{
    short=DARE,
    long=discrete-time algebraic Riccati equation,
}
\DeclareAcronym{qos}{
    short=QoS,
    long=quality-of-service
}
\DeclareAcronym{lqg}{
    short=LQG,
    long=linear-quadratic-Gaussian
}
\DeclareAcronym{socp}{
    short=SOCP,
    long=second-order cone programming 
}
\DeclareAcronym{sca}{
    short=SCA,
    long=successive convex approximation 
}
\DeclareAcronym{cdf}{
    short=CDF,
    long=cumulative distribution function 
}
\DeclareAcronym{ci}{
    short=CI,
    long=confidence interval
}
\DeclareAcronym{zf}{
    short=ZF,
    long=zero-forcing
}
\DeclareAcronym{dofs}{
    short=DoFs,
    long=degrees of freedom
}
\DeclareAcronym{uav}{
    short=UAV,
    long=unmanned aerial vehicle 
}
\DeclareAcronym{iiot}{
    short=IIoT,
    long=industrial internet of things
}
\DeclareAcronym{isac}{
    short=ISAC,
    long=integrated sensing and communication
}
\DeclareAcronym{crb}{
    short=CRB,
    long=Cramér-Rao bound
}
\DeclareAcronym{wrt}{
    short=w.r.t.,
    long=with respect to
}

\begin{document}
\title{Joint Communication and Control Beamforming: A Closed-Loop Control Perspective}

\author{Hao Jiang, Chongjun Ouyang, Yuanwei Liu,~\IEEEmembership{Fellow,~IEEE}, Arumugam Nallanathan,~\IEEEmembership{Fellow,~IEEE}, 
\\ and Robert Schober,~\IEEEmembership{Fellow,~IEEE}
\thanks{H. Jiang, C. Ouyang, and A. Nallanathan are with the School of Electronic Engineering and Computer Science, Queen Mary University of London, London, E1 4NS, U.K. (email: \{hao.jiang, c.ouyang, a.nallanathan\}@qmul.ac.uk).}
\thanks{Y. Liu is with the Department of Electrical and Electronic Engineering, The University of Hong Kong, Hong Kong (email: yuanwei@hku.hk).}
\thanks{R. Schober is with the Institute for Digital Communications, Friedrich-Alexander-University Erlangen-Nürnberg (FAU), Germany (e-mail: robert.schober@fau.de).
}}
\maketitle

\begin{abstract}
A joint communication and control (JCC) framework is proposed, in which a base station (BS) simultaneously serves multiple communication users (CUs) and controls a physical plant in a closed-loop manner:
i) In the downlink, control inputs generated by the BS are transmitted and recovered at the plant, where the resulting wireless actuation distortion is incorporated into the plant-state evolution; 
ii) In the uplink, the current state of the plant is reported to the BS and tracked using a Kalman filter (KF) to support subsequent control-input generation.
To characterize the long-term control performance in the presence of inter-function interference, finite- and infinite-horizon linear quadratic Gaussian (LQG) costs are derived, which directly relate the beamforming design to the long-term evolution of the plant state.
Building on the above, JCC beamforming problems are formulated for both vector- and scalar-valued control inputs, in which the infinite-horizon LQG cost is minimized subject to per-user communication signal-to-interference-plus-noise ratio (SINR) requirements.
For the vector case, a second-order cone programming (SOCP)-based successive convex approximation method is developed to address the resulting nonconvex problem.
For the scalar case, a closed-form expression for the infinite-horizon LQG cost is derived, based on which the communication-control Pareto boundary is optimally characterized by an SOCP-based bisection method. 
The optimality of this method is guaranteed by proving the strict monotonicity of the scalar-valued control cost with respect to the control SINR.
Finally, our numerical results reveal: i) The derived control costs closely match those obtained via Monte Carlo simulations, while the KF method accurately tracks the ground-truth plant-state trajectory; and ii) the proposed methods consistently outperform the zero-forcing-based benchmark, indicating that balancing communication--control interference is beneficial, particularly when the available spatial degrees of freedom (DoFs) are limited.
\end{abstract}

\section{Introduction}
Over the past decades, \ac{mimo} technology has played a pivotal role in the evolution of wireless communication systems by providing substantial multiplexing gains \cite{heath2026mimo}. 
In particular, deploying multiple antennas at the transceivers provides additional spatial \ac{dofs}, thereby enabling the simultaneous transmission of parallel data streams over the same time-frequency resources \cite{chen2021massive, emil2025enabling}.
In addition to enhancing the spectral efficiency, the spatial \ac{dofs} offered by \ac{mimo} can also facilitate the sharing of a common hardware infrastructure among multiple functionalities \cite{wen2025survey}. 
By exploiting these spatial \ac{dofs} to suppress both inter-user and inter-function interference, \ac{mimo} technology enables communication to be integrated with sensing, computing, and caching. 
Despite extensive research on multifunctional \ac{mimo} systems, the integration of communication and control has received comparatively limited attention.

Nevertheless, \ac{jcc} is becoming increasingly important for emerging cyber-physical applications, such as the \ac{iiot}, \ac{uav} control, and robotic networks \cite{liu2026b2xnetworks}. 
In these scenarios, wireless channels facilitate both information transmission and the delivery of control inputs that directly influence physical processes.
Specifically, the controller sends control inputs to the physical plant, which then reports its updated state to facilitate subsequent state estimation and the generation of new control inputs.
Under this closed-loop architecture, communication impairments can not only degrade the received-signal quality but also distort the applied control input, thereby perturbing the plant state and affecting subsequent control decisions. 
As demonstrated in \cite{zeng2019joint} and \cite{chang2022joint}, communication throughput and control stability are tightly coupled through the wireless channel, necessitating dedicated designs that jointly coordinate the two functionalities.

Despite its importance, \ac{jcc} fundamentally differs from other integrated-functionality paradigms, because control is a dynamic and temporally coupled process rather than an instantaneous functionality \cite{tatikonda2004control, leong2011optimal, sahai2006necessity}. 
Specifically, in \ac{isac} systems, sensing performance is typically evaluated using instantaneous or frame-level metrics, including sensing \ac{snr}, beampattern gain, and \ac{crb} \cite{liu2022survey_isac}. 
Similarly, integrated communication and computing or caching is typically characterized by metrics, such as computation latency, energy consumption, and cache-hit probability \cite{adam2023beyond}. 
These metrics also depend on short-term communication-resource allocation and generally do not capture the temporal evolution of a physical state over a closed-loop process.
By contrast, the control inputs in \ac{jcc} systems are generated by the controller and transmitted over wireless channels to actuate the plant repeatedly until the closed-loop system reaches its steady-state. 
Hence, any distortion of the transmitted control inputs can perturb the plant state, thereby affecting subsequent observations, state estimates, control actions, and ultimately the long-term control performance. 
Consequently, instantaneous metrics directly borrowed from other integrated-functionality systems are generally insufficient to fully characterize the performance of \ac{jcc} systems.

To bridge this gap, several studies have investigated the interplay between communication and control.
One line of research characterizes application-specific communication-control coupling and the resulting trade-offs.
In particular, the authors of \cite{zeng2019joint} derived the wireless-delay constraint for vehicular-platoon stability and optimized the control parameters to enhance system reliability.
Moreover, the authors of \cite{chang2022joint} investigated the relationship between the control-state error and beam-misalignment error, based on which an event-triggered control-update policy was developed.
Furthermore, the authors of \cite{gan2026modeling} characterized the trade-off between communication delay and steady-state control variance, and computed the corresponding performance regions and outage probabilities.
In parallel, another line of research incorporates control requirements into physical-layer transmission optimization.
More specifically, the authors of \cite{lyu2018dynamics} quantified the maximum tolerable packet-loss rate based on a predefined Lyapunov function, based on which a beamforming-assisted control-input scheduling scheme was devised.
Likewise, the authors of \cite{wang2024stability} considered receive-beamforming design subject to Lyapunov-based stability constraints, while the authors of \cite{ma2025robust} proposed a robust transmit-beamforming design subject to probabilistic successful-transmission constraints.
Moreover, the authors of \cite{leong2011optimal} studied the \ac{lqg} control process over a point-to-point fading channel, thereby establishing the fundamental closed-loop control framework for \ac{jcc}.
Lastly, from an information-theoretic perspective, the authors of \cite{sabag2023reducing} characterized the minimum conditional directed information required to achieve a prescribed \ac{lqg} cost.

However, most existing studies address only specific components of the closed-loop control process or model wireless transmission using high-level abstractions.
Specifically, commonly adopted reliability metrics and stability criteria, such as packet-loss probability, successful-transmission probability, and control stability, do not directly reveal how wireless-induced control-input distortions affect plant-state evolution, state estimation, and subsequent control actions.
Although the authors of \cite{sabag2023reducing} directly considered the \ac{lqg} cost, their approach represented communication solely by the quantity of the information transmitted from the transmitter to the controller.
Consequently, this abstraction does not allow the framework to characterize the effects of the channel conditions, noise, interference, or beamforming on closed-loop control performance.
Moreover, the authors of \cite{gan2026modeling} considered a \ac{jcc} system with one \ac{cu} and one plant, where state-reporting and control-input reconstruction errors were modeled using rate-distortion relationships rather than being directly derived from the received signals.
Furthermore, the point-to-point \ac{jcc} framework in \cite{leong2011optimal} was limited to a simplified single-control-loop scenario.
Consequently, the direct relationship between beamforming and long-term closed-loop control performance remains insufficiently characterized in multi-\ac{cu} systems with vector-valued control inputs.
Hence, a more comprehensive control metric is needed to directly connect beamforming to control-input recovery and plant-state evolution, while accounting for plant-state deviation, applied control effort, and the state-estimation error associated with the uplink state-feedback process.
Only with such a metric can communication and control functionalities be jointly designed to address both inter-user and inter-function interference while minimizing the long-term \ac{lqg} cost.
Motivated by the above, the main contributions of this paper can be summarized as follows:
\begin{itemize}
    \item We propose a \ac{mimo} \ac{jcc} framework in which a multi-antenna \ac{bs} simultaneously serves multiple \acp{cu} and transmits a vector-valued control input to a plant. 
    In the proposed system, the interference between the communication and control functionalities is explicitly modeled. 
    To establish the closed-loop control architecture, \ac{lmmse}-based downlink control-input recovery at the plant and \ac{kf}-based uplink state tracking at the \ac{bs} are incorporated.
    
    \item We derive the optimal finite- and infinite-horizon \ac{lqg} costs for the considered wireless closed-loop control system. 
    In contrast to conventional control metrics that abstract communication, the derived costs explicitly capture the wireless actuation distortion caused by interference and noise, as well as state-estimation errors in the feedback link. 
    Hence, the dependence of long-term control performance on beamforming is directly captured.

    \item For the vector control-input case, we formulate a \ac{jcc} beamforming problem to minimize the infinite-horizon \ac{lqg} cost subject to the communication \ac{qos} requirements and the transmit-power constraint. 
    Given the inherent nonconvexity of this optimization problem, an \ac{socp}-based \ac{sca} algorithm is developed with backtracking to obtain a suboptimal solution.

    \item For the scalar control-input case, we derive a closed-form expression for the infinite-horizon \ac{lqg} cost. 
    To characterize the communication-control trade-off, a simplified system with a single \ac{cu} and a single plant is considered. 
    In this setting, a \ac{jcc} beamforming problem is formulated to minimize the infinite-horizon control cost subject to the communication \ac{qos} requirement and the transmit-power constraint. 
    By exploiting the monotonicity of the infinite-horizon \ac{lqg} cost \ac{wrt} the control \ac{sinr}, this problem is optimally solved using an \ac{socp}-based bisection approach.

     \item Numerical results demonstrate: i) The derived \ac{lqg} costs for both the vector and scalar control-input cases closely match the corresponding Monte Carlo results; ii) the proposed \ac{kf}-based method accurately tracks the plant-state trajectory; and iii) the proposed algorithms consistently outperform the heuristic \ac{zf}-based benchmark, demonstrating that balancing inter-function interference is necessary, especially when the available spatial \ac{dofs} are limited.
\end{itemize}
\emph{Organization:} The rest of this paper is organized as follows.
Section~\ref{sect:system_model} presents the \ac{jcc} system model and characterizes the finite- and infinite-horizon \ac{lqg} costs.
Section~\ref{sect:beamforming_design_jcc_vector_valued} develops the beamforming design for vector-valued control inputs, while Section~\ref{sect:beamforming_design_jcc_scalar_valued} considers the scalar control-input case and characterizes the communication-control Pareto boundary.
Numerical results are provided in Section~\ref{sect:numerical_results}, and
conclusions are drawn in Section~\ref{sect:conclusions}.

\emph{Notations:} Throughout this paper, scalars are denoted by italic letters, while vectors and matrices are denoted by boldface lowercase and uppercase letters, respectively.
$\mathbb{C}^{M\times N}$ and $\mathbb{R}^{M \times N}$ denote the spaces of $M\times N$ complex- and real-valued matrices, respectively. 
$(\cdot)^{\textsf{T}}$, $(\cdot)^{\textsf{H}}$, $(\cdot)^{\mathscr{*}}$, and $(\cdot)^{-1}$ denote the transpose, conjugate-transpose, conjugate, and inverse operations, respectively.
$[\cdot]_{i,:}$, $[\cdot]_{:, j}$, and $[\cdot]_{i,j}$ denote the operations to extract the $i$-th row, the $j$-th column, and the $(i,j)$-th entry of a matrix, respectively.
Here, $\nabla$ denotes the gradient operator, and $\mathrm{vec}\{\cdot\}$ denotes the vectorization operator.
$\mathrm{j}=\sqrt{-1}$ and $\mathrm{e}$ denote the imaginary unit and Euler's number, respectively.
$\mathbf{0}_{M}$ and $\mathbf{0}_{M \times N}$ denote a $M \times 1$ zero vector and a $M \times N$ zero matrix, respectively. 
$\mathbb{E}[\cdot]$ represents statistical expectation, while $\Re\{\cdot\}$ and $\Im\{\cdot\}$ denote the real and imaginary parts of a complex number, respectively.
$\mathcal{CN}(\mu, \sigma^2)$ and $\mathcal{N}(\mu, \sigma^2)$ denote the complex Gaussian distribution and the normal distribution with mean $\mu$ and variance $\sigma^2$, respectively.
$\| \cdot \|_2$, $\| \cdot \|_F$, and $|\cdot|$ denote the $L_2$ norm, the Frobenius norm, and the absolute value, respectively.
The operator $\angle (\cdot)$ represents the extraction of the phase of a complex number.
Finally, $\mathcal{O}(\cdot)$ denotes the Big-O notation.

\section{System Model} \label{sect:system_model}
\begin{figure}[!t]
    \centering
    \includegraphics[width=0.85\linewidth]{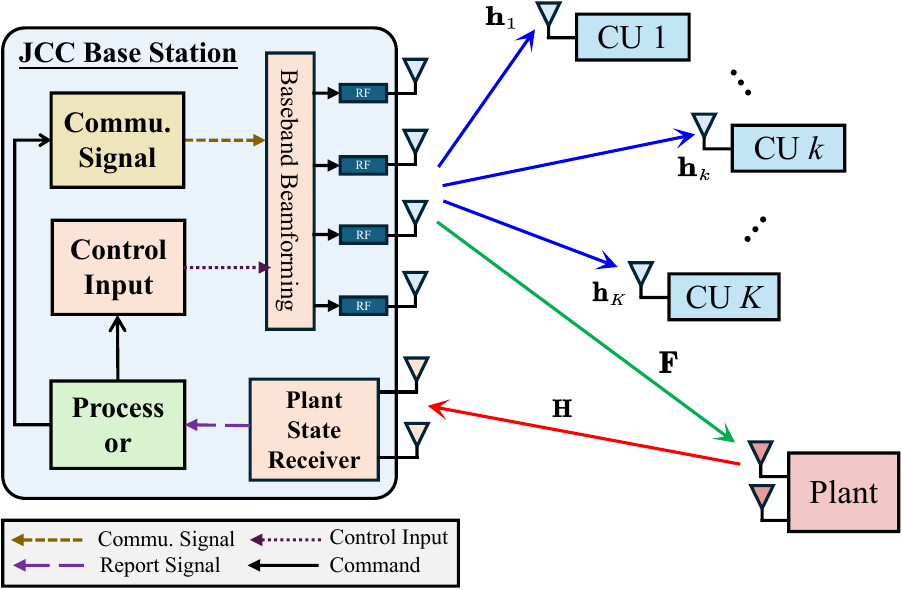} 
    \caption{Illustration of the JCC system model.}
    \label{fig:system_model}
\end{figure}
As illustrated in Fig. \ref{fig:system_model}, we consider a \ac{jcc} system, where $K$ single-antenna \acp{cu} are served in the downlink and a plant with an $M_{\rm p}$-antenna array is controlled by the \ac{bs}. 
To simultaneously transmit and receive, the \ac{bs} is equipped with an $M_{\rm t}$-antenna array for downlink transmission and an $M_{\rm r}$-antenna array for uplink reception.
In particular, the downlink signals comprise the communication symbols for the \acp{cu} and the control input to the plant, while the uplink transmission delivers the plant-state report from the plant to the BS.
For the communication link, the channel vector from the \ac{bs} to the $k$-th \ac{cu} is denoted by $\mathbf{h}_k \in \mathbb{C}^{M_{\rm t} \times 1}$, where $k\in\mathcal{K}\triangleq\{1,2,\dotsc, K\}$.
For the control link, the channel matrix from the \ac{bs} to the plant is denoted by $\mathbf{F} \in \mathbb{C}^{M_{\rm p} \times M_{\rm t}}$.
For the state-feedback link, the channel matrix from the plant to the \ac{bs} is denoted by $\mathbf{H}\in\mathbb{C}^{M_{\mathrm{r}} \times M_{\mathrm{p}}}$.
Prior to system modeling, we summarize the core assumptions used in this paper: i) The wireless channels are assumed to be quasi-static within each considered control horizon, and perfect \ac{csi} is available at both the \ac{bs} and the plant \cite{yang2014quasi,luo2015capacity}; ii) communication symbols, receiver noises, control-process noise, and the initial plant state are mutually independent, and all time-varying quantities are independent across time; iii) the uplink and downlink transmissions utilize orthogonal resources, thereby avoiding mutual interference \cite{lin2019joint}; and iv) the communication and control signals are zero-mean.

\subsection{Communication Model}
Let the time slot be indexed by $n$, where $n\in \mathcal{N}\triangleq\{1, \ldots, N\}$.
For each time slot, the \ac{bs} simultaneously transmits Gaussian communication symbols $s_{k, n} \sim \mathcal{CN}(0, 1)$ to the $K$ \acp{cu} and the control input $\mathbf{u}_n \in \mathbb{R}^{L \times 1}$ to the plant.\footnote{Following standard linear-control models \cite{chen2016optimal}, the control input is modeled as a real-valued vector. 
Alternatively, two real-valued entries of the control-input vector can be mapped to the in-phase and quadrature dimensions of a single complex symbol for more efficient transmission.
This extension is left for future work. }
Since the original control-input vector $\mathbf{u}_n \in \mathbb{R}^L$ generally does not have an identity covariance matrix, it is first normalized as $\tilde{\mathbf{u}}_n = \boldsymbol{\Pi}^{-1/2}_n\mathbf{u}_n$, where $\boldsymbol{\Pi}_n \triangleq \mathbb{E}\left[\mathbf{u}_n\mathbf{u}^\mathsf{T}_n\right]$.
Here, $\boldsymbol{\Pi}_n$ is obtained from the long-term second-order statistics of the control input and is treated as known during beamforming design~\cite{chen2016optimal,gattami2014multiobjective}.
The subscript $n$ indicates possible time variations of these statistics across a finite horizon. 
In the stationary infinite-horizon regime, $\boldsymbol{\Pi}_n$ converges to a stationary value $\boldsymbol{\Pi}$ as $n \rightarrow \infty$.
The resulting normalized control input $\tilde{\mathbf{u}}_n$ satisfies $\mathbb{E}\left[\tilde{\mathbf{u}}_n\tilde{\mathbf{u}}^{\mathsf{T}}_n\right] = \mathbf{I}_{L}$, so that the transmit power allocated to the control-input transmission is determined solely by the corresponding precoder.

Let $\mathbf w_{k} \in \mathbb C^{M_{\rm t}\times 1}$ denote the beamforming vector for the $k$-th CU, and let $\mathbf W_{\mathrm{p}}\in\mathbb C^{M_{\rm t}\times L}$ denote the precoding matrix for the control input.
Here, for valid spatial multiplexing and control-input recovery, we assume that $K \le M_{\rm t}$, $L \le M_{\rm p}$, and $L \le M_{\rm r}$ hold, and the corresponding effective channels have sufficient rank for signal recovery.
Consequently, the transmit signal can be expressed as follows:
\begin{align}
   \mathbf{s}_n \triangleq \sum\nolimits_{k=1}^K{\mathbf{w}_{k} s_{k,n}}+\mathbf{W}_{\mathrm{p}}\tilde{\mathbf{u}}_n,
\end{align}
where the overall transmit power satisfies $\sum\nolimits_{k=1}^K{\left\| \mathbf{w}_k \right\| _{2}^{2}}+\left\| \mathbf{W}_{\mathrm{p}} \right\| _{F}^{2}\le P_{\max}$, and $P_{\max}$ denotes the total transmit power at the \ac{bs}.
For the $k$-th \ac{cu}, the received signal is given by
\begin{align}
    y_{k,n} &=\mathbf{h}_{k}^{\mathsf{H}}\mathbf{s}_n +z_{k, n} =\mathbf{h}_{k}^{\mathsf{H}}\mathbf{w}_{k}s_{k,n} \notag \\
    &+\underset{\mathrm{Commun}.~\mathrm{Interf.}}{\underbrace{\sum\nolimits_{j=1,j\ne k}^K{\mathbf{h}_{k}^{\mathsf{H}}\mathbf{w}_{j}s_{j,n}}}}+\underset{\mathrm{Control}~\mathrm{Interf.}}{\underbrace{\mathbf{h}_{k}^{\mathsf{H}}\mathbf{W}_{\mathrm{p}}\tilde{\mathbf{u}}_n }}+z_{k,n}, \label{eq:received_signal}
\end{align}
where $z_{k,n} \sim \mathcal{CN}(0, \sigma_{\mathrm{c}, k}^2)$ denotes additive Gaussian noise with variance $\sigma_{\mathrm{c}, k}^2$.
For brevity, we let $\sigma_{\mathrm{c},k}^2 = \sigma_{\mathrm{c}}^2$, $\forall k$.
Accordingly, the \ac{sinr} at CU $k$ can be expressed as follows:
\begin{align}
    &\gamma _k (\mathbf{W}_{\mathrm{c}}, \mathbf{W}_{\mathrm{p}})=\frac{\left| \mathbf{h}_{k}^{\mathsf{H}}\mathbf{w}_{k} \right|^2}{\sum\nolimits_{j=1,j\ne k}^K{\left| \mathbf{h}_{k}^{\mathsf{H}}\mathbf{w}_{j} \right|^2}+\left\| \mathbf{h}_{k}^{\mathsf{H}}\mathbf{W}_{\mathrm{p}} \right\|_2 ^2+\sigma _{\mathrm{c}}^{2}}, \label{eq:sinr}
\end{align}
where $\mathbf{W}_{\mathrm{c}} \triangleq [\mathbf{w}_{1}, \dots, \mathbf{w}_{K}] \in \mathbb{C}^{M_{\rm t} \times K}$ collects the beamforming vectors for the $K$ \acp{cu}.
The communication \ac{sinr} is used as the performance metric for the communication functionality.
As shown in \eqref{eq:received_signal} and \eqref{eq:sinr}, the interference experienced by each \ac{cu} is caused by the signals intended for the other \acp{cu} and by the control signal intended for the plant.

\subsection{Control Model} \label{subsect:vector_control_input}
For the plant, there are two transmission directions: The downlink control-input transmission from the \ac{bs} to the plant and the uplink state-feedback transmission from the plant to the \ac{bs}.
In this subsection, we first describe the transmission of the control input and its effect on the plant-state evolution. 
Next, we explain how the updated plant state is reported to and tracked by the \ac{bs} for subsequent control-input generation. 
Finally, we introduce the control-performance metric.

\subsubsection{Downlink Control-Input Transmission}
Let $\mathbf{z}_n \sim \mathcal{CN}(\boldsymbol{0}_{M_{\rm p}}, \sigma_{\rm p}^2\mathbf{I}_{M_{\mathrm{p}}})$ be the additive Gaussian noise at the plant.
The received downlink signal at the plant during the $n$-th time slot can be expressed as follows:
\begin{align}
    \mathbf{r}_n&=\mathbf{F}\mathbf{s}_n+\mathbf{z}_n=\mathbf{F}\mathbf{W}_{\mathrm{p}}\tilde{\mathbf{u}}_n + \underset{\mathrm{Commun}.~\mathrm{Interf}.}{\underbrace{\sum\nolimits_{k=1}^K{\mathbf{F}\mathbf{w}_{k}s_{k,n}}}} +\mathbf{z}_n\notag \\ &=\mathbf{F}\mathbf{W}_{\mathrm{p}}\tilde{\mathbf{u}}_n+\mathbf{d}_n, \label{eq:plant_observation_vector}
\end{align}
where $\mathbf{d}_n\triangleq \sum\nolimits_{k=1}^K{\mathbf{F}\mathbf{w}_{k}s_{k,n}} + {\mathbf{z}}_n$ collects the interference and noise.
Under the zero-mean assumption, the mean of $\mathbf{d}_n$ is given by $\mathbb{E}\left[\mathbf{d}_n \right] = \boldsymbol{0}_{M_{\mathrm{p}}}$.
Based on the statistical independence assumption, the covariance matrix of $\mathbf{d}_n$ can be computed as $\mathbf{R}_{d} \triangleq \mathbb{E}[\mathbf{d}_n \mathbf{d}_n^{\mathsf{H}}]=\sum\nolimits_{k=1}^K{\mathbf{F}} \mathbf{W}_{k}\mathbf{F}^{\mathsf{H}}+\sigma _{\mathrm{p}}^{2}\mathbf{I}_{M_{b\mathrm{p}}}$, where $\mathbf{W}_{k} \triangleq \mathbf{w}_{k}^{}\mathbf{w}_{k}^{\mathsf{H}}$.
Since the control-input vector is real-valued, we stack the real and imaginary parts of the received signal in \eqref{eq:plant_observation_vector}.
The resulting real-valued expressions are given by
\begin{equation}
\begin{aligned}
     &\breve{\mathbf{r}}_n\triangleq \left[ \begin{array}{c}
	\Re \left\{ \mathbf{r}_n \right\}\\
	\Im \left\{ \mathbf{r}_n \right\}\\
\end{array} \right]\in \mathbb{R} ^{2M_{\mathrm{p}}\times 1}, ~~ \breve{\mathbf{d}}_n\triangleq \left[ \begin{array}{c}
	\Re \left\{ \mathbf{d}_n \right\}\\
	\Im \left\{ \mathbf{d}_n \right\}\\
\end{array} \right] \in \mathbb{R} ^{2M_{\mathrm{p}}\times 1} \\
    &\breve{\mathbf{F}}\triangleq \left[ \begin{array}{c}
	\Re \left\{ \mathbf{FW}_{\mathrm{p}} \right\}\\
	\Im \left\{ \mathbf{FW}_{\mathrm{p}} \right\}\\
\end{array} \right] \in \mathbb{R} ^{2M_{\mathrm{p}}\times L},   \\
    &\breve{\mathbf{R}}_{d}\triangleq\frac{1}{2}\left[ \begin{matrix}
	\Re \left\{ \mathbf{R}_{d} \right\}&		-\Im \left\{ \mathbf{R}_{d} \right\}\\
	\Im \left\{ \mathbf{R}_{d} \right\}&		\Re \left\{ \mathbf{R}_{d} \right\}\\
\end{matrix} \right] \in \mathbb{R} ^{2M_{\mathrm{p}}\times 2M_{\mathrm{p}}}. 
\end{aligned}\label{eq:real_valued_definitions}
\end{equation}
Here, the factor $1/2$ in $\breve{\mathbf{R}}_{d}$ arises, since the variance of the complex Gaussian disturbance is equally divided between its real and imaginary parts.
Applying these definitions, the received control input at the plant can be rewritten as follows:
\begin{align}
    \breve{\mathbf{r}}_n=\breve{\mathbf{F}} \tilde{\mathbf{u}}_n+\breve{\mathbf{d}}_n.
\end{align}
To extract the normalized control input $\tilde{\mathbf{u}}_n$, the \ac{lmmse} combiner $\mathbf{G} = \breve{\mathbf{F}}^{\mathsf{T}}(\breve{\mathbf{F}}\breve{\mathbf{F}}^{\mathsf{T}} + \breve{\mathbf{R}}_{d})^{-1}$ is applied.
Therefore, the estimated normalized control input $\hat{\tilde{\mathbf{u}}}_n$ is given by 
\begin{align}
    \hat{\tilde{\mathbf{u}}}_n=\mathbf{G}\breve{\mathbf{F}}\tilde{\mathbf{u}}_n+\mathbf{G}\breve{\mathbf{d}}_n=\mathbf{M}\tilde{\mathbf{u}}_n+\boldsymbol{\vartheta }_n, \label{eq:estimated_normalized_control_input}
\end{align}
where $\mathbf{M}\triangleq \mathbf{G}\breve{\mathbf{F}}$ denotes the normalized recovery matrix and $\boldsymbol{\vartheta }_n\triangleq \mathbf{G}\breve{\mathbf{d}}_n$ denotes the additive actuation disturbance induced by the interference and receiver noise.
The covariance of this disturbance is defined as $\mathbb{E} \left[ \boldsymbol{\vartheta }_{n}\boldsymbol{\vartheta }_{n}^{\mathsf{T}} \right] = \mathbf{G}\breve{\mathbf{R}}_{d}\mathbf{G}^{\mathsf{T}}$.
Recall that $\tilde{\mathbf{u}}_n = \boldsymbol{\Pi}^{-1/2}_n\mathbf{u}_n$.
The estimate of the original control input is recovered by rescaling its normalized estimate in \eqref{eq:estimated_normalized_control_input}.
Hence, the resulting estimated control input can be expressed as follows:
\begin{align}
    \hat{\mathbf{u}}_n&=\mathbf{\Pi }_{n}^{1/2}\hat{\tilde{\mathbf{u}}}_n \notag \\
    &=\mathbf{\Pi }_{n}^{1/2}\mathbf{M}\mathbf{\Pi }_{n}^{-1/2}\mathbf{u}_n+\mathbf{\Pi }_{n}^{1/2}\boldsymbol{\vartheta }_n=\bar{\mathbf{M}}_n\mathbf{u}_n+\mathbf{e}_n, \label{eq:received_control_command_signal}
\end{align}
where $\bar{\mathbf{M}}_n\triangleq \boldsymbol{\Pi }_{n}^{1/2}\mathbf{M\Pi }_{n}^{-1/2}$ denotes the denormalized recovery matrix and $\mathbf{e}_n\triangleq \mathbf{\Pi }_{n}^{1/2}\boldsymbol{\vartheta }_n$ denotes the additive actuation disturbance after the denormalization.
The covariance of $\mathbf{e}_n$ is characterized by $\mathbf{\Sigma }_{e,n}\triangleq \mathbb{E} \left[ \mathbf{e}_n\mathbf{e}_{n}^{\mathsf{T}} \right] =\mathbf{\Pi }_{n}^{1/2}\mathbf{G}\breve{
\mathbf{R}}_{d}\mathbf{G}_{}^{\mathsf{T}}\mathbf{\Pi }_{n}^{1/2}$, while its mean is given by $\mathbb{E}[\mathbf{e}_n]=\boldsymbol{0}_L$.

Once $\hat{\mathbf{u}}_n$ is obtained, the plant updates its state.
According to \cite{sahai2006necessity} and \cite{sahai2007necessity}, we model the plant control process as a discrete-time vector-valued stochastic linear system.
More specifically, when the plant applies the recovered control input $ \hat{\mathbf{u}}_n$, the state-evolution equation from time slot $n$ to time slot $n+1$ is given by
\begin{align}
    \mathbf{x}_{n+1} = \mathbf{A}\mathbf{x}_n + \mathbf{B} \hat{\mathbf{u}}_n + \mathbf{v}_n, \label{eq:control_equation}
\end{align}
where $\mathbf{x}_{n} \in \mathbb{R}^{L \times 1}$ represents the plant state at time slot $n$, $\mathbf{A} \in \mathbb{R}^{L \times L}$ and $ \mathbf{B} \in \mathbb{R}^{L \times L}$ are the real-valued constant matrices that characterize how the plant state evolves in response to different actions, and $\mathbf{v}_n \sim \mathcal{N}(\boldsymbol{0}_L, \boldsymbol{\Sigma}_v)$ denotes the additive Gaussian control-process noise.
The covariance of the noise term is given by $\boldsymbol{\Sigma}_v = \mathbb{E} \left[ \mathbf{v}_n\mathbf{v}_{n}^{\mathsf{T}} \right] =\sigma _{v}^{2}\mathbf{I}_L$.
According to \eqref{eq:received_control_command_signal}, the state-evolution equation can be further expressed as follows:
\begin{align}
    \mathbf{x}_{n+1} = \mathbf{A}\mathbf{x}_n + \mathbf{B}\bar{\mathbf{M}}_n{\mathbf{u}}_n + \mathbf{B}\mathbf{e}_n + \mathbf{v}_n. \label{eq:status_evolution}
\end{align}
\begin{remark}
    (Factors Affecting the Control Process) \emph{Based on \eqref{eq:status_evolution}, the factors affecting the control process fall into two categories: 
    \textbf{i) Transmission Imperfections}: Due to the presence of \acp{cu}, the communication signals can interfere with the control-input transmission, as indicated by \eqref{eq:plant_observation_vector}. 
    In addition, the additive Gaussian noise at the receiver further distorts the received control inputs.
    \textbf{ii) Control-Model Imperfections}: Due to plant uncertainty and unmodeled dynamics, the plant state evolution is also affected by the process noise originating from sources distinct from the additive noise on the communication side. }
\end{remark}
\begin{figure*}[t!]
    \centering
    \includegraphics[width=0.65\linewidth]{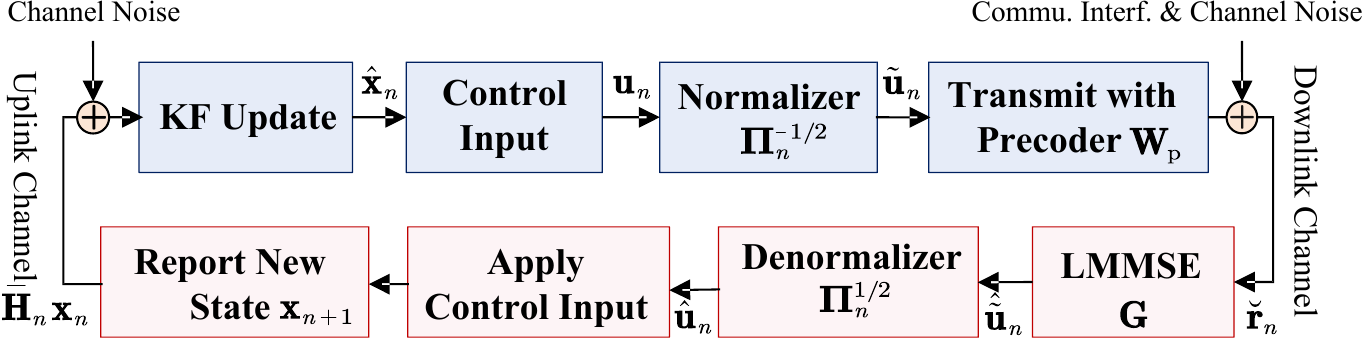}
    \caption{Illustration of the closed-loop control process in the JCC system.}
    \label{fig:close_loop_control}
\end{figure*}
\subsubsection{Uplink Plant-State Reporting} \label{sect:uplink_modeling}
To generate appropriate control inputs for the upcoming time slot, the \ac{bs} needs to know the plant's updated state $\mathbf{x}_{n+1}$, which is obtained from the uplink state-feedback mechanism.
According to \eqref{eq:control_equation}, the updated plant state is $\mathbf{x}_{n+1}$ after $\hat{\mathbf{u}}_n$ has been applied.
Similar to the control input transmission, the plant state is first normalized as $\tilde{\mathbf{x}}_{n+1} = \boldsymbol{\Omega}_{n+1}^{-1/2} \mathbf{x}_{n+1}$, where $\boldsymbol{\Omega}_{n+1} \triangleq \mathbb{E}\left[ \mathbf{x}_{n+1} \mathbf{x}_{n+1}^{\mathsf{T}}\right]$.
Similar to $\boldsymbol{\Pi}_n$, $\boldsymbol{\Omega}_{n+1}$ is treated as known during beamforming design, and may vary across time slots in the finite-horizon case, while $\boldsymbol{\Omega}_{n}$ converges to a stationary value $\boldsymbol{\Omega}$ as $n \rightarrow \infty$.
Then, the plant transmits the normalized state vector $\tilde{\mathbf{x}}_{n+1}$ to the \ac{bs} over the uplink channel $\mathbf{H}$ using the precoding matrix $\mathbf{U}$.

Hence, the received signal at the \ac{bs} is given by
\begin{align}
    \mathbf{c}_{n+1}=\bar{\mathbf{H}}_{n+1}\mathbf{x}_{n+1}+\mathbf{n}_{n+1}, \label{eq:uplink_state_report}
\end{align}
where $\bar{\mathbf{H}}_{n+1}\triangleq \sqrt{P_{\rm p}}\mathbf{H}\mathbf{U}\mathbf{\Omega }_{n+1}^{-1/2}$, $P_{\rm p}$ denotes total transmit power at the plant, and $\mathbf{n}_{n+1} \sim \mathcal{CN}(\boldsymbol{0}_{M_{\mathrm{r}}}, \sigma_{\rm f}^2 \mathbf{I}_{M_{\mathrm{r}}})$ denotes the additive Gaussian noise.
In this work, we adopt \ac{svd}-based eigenmode precoding, i.e., $\mathbf{U}=L^{-1/2}\left[ \mathbf{V}_{\mathrm{t}} \right] _{:,1:L}$.
Here, $\mathbf{H}=\mathbf{V}_{\mathrm{r}}\mathbf{\Lambda V}_{\mathrm{t}}^{\mathsf{H}}$ denotes the \ac{svd} of $\mathbf{H}$, where the singular values are arranged in descending order, and the corresponding singular vectors are ordered accordingly.
According to the adopted link-separation assumption, the uplink state-reporting transmission does not interfere with the downlink transmission.

To track the plant's state across time slots, a \ac{kf} approach is adopted at the \ac{bs} \cite{leong2011optimal}.
After receiving the uplink state report $\mathbf{c}_{n+1}$, the \ac{bs} updates the plant-state estimate and obtains the posterior estimate denoted by $\hat{\mathbf{x}}_{n+1} \triangleq \hat{\mathbf{x}}_{n+1\mid n+1}$.
The estimated state $\hat{\mathbf{x}}_{n+1}$ is subsequently utilized to generate the control input $\mathbf{u}_{n+1}$.
The detailed \ac{kf} implementation is provided in Appendix \ref{appendix:implementation_of_kf}.
For ease of understanding, Fig. \ref{fig:close_loop_control} summarizes the closed-loop control process.

To quantify the uncertainty caused by the \ac{kf} module, we define the \ac{kf} state-estimation error in what follows. 
Let $n$ denote a general time slot. 
The corresponding state-estimation error is defined as $\boldsymbol{\epsilon}_n \triangleq \mathbf{x}_n - \hat{\mathbf{x}}_n$. 
Its covariance matrix is given by $\mathbf{C}_n \triangleq \mathbb{E}[\boldsymbol{\epsilon}_n\boldsymbol{\epsilon}_n^{\mathsf{T}}]$.
Based on the orthogonality property of the \ac{kf} estimate, the estimation error is uncorrelated from the estimated state, i.e., $\mathbb{E}[\boldsymbol{\epsilon}_n\hat{\mathbf{x}}_n^{\mathsf{T}}]$ is a zero matrix.
Therefore, the relationship between the true plant state and the estimated plant state is given by
\begin{align}
    \mathbf{x}_n = \hat{\mathbf{x}}_n + \boldsymbol{\epsilon}_n, \notag 
\end{align}
where the true plant state can be viewed as the \ac{kf} estimate perturbed by an additive estimation error, capturing the state uncertainty remaining after \ac{kf}-based estimation.

\subsubsection{Performance Metric of Control}
On the control side, the objective of the \ac{bs} is to ensure stable plant operation.
To quantify the long-term control performance, we leverage the \ac{lqg} framework \cite{van1981certainty}.
For a finite horizon of $N$ time slots, the finite-horizon cost is defined as follows:
\begin{align}
    &J_N(\mathbf{W}_{\mathrm{c}},\mathbf{W}_{\mathrm{p}})=\notag \\
    &\qquad \quad \mathbb{E} \left[ \sum\nolimits_{n=1}^{N-1}{\left( \mathbf{x}_{n}^{\mathsf{T}}\mathbf{Qx}_{n}^{}+\hat{\mathbf{u}}_{n}^{\mathsf{T}}\mathbf{D}\hat{\mathbf{u}}_{n}^{} \right) +\mathbf{x}_{N}^{\mathsf{T}}\mathbf{Sx}_{N}^{}} \right], \label{eq:control_cost_finite}
\end{align}
where $\mathbf{Q} \succeq 0$ penalizes the deviation of the physical state from the desired equilibrium, which is set to $\mathbf{0}_L$, $\mathbf{D} \succeq 0 $ penalizes the actually executed control inputs, i.e., the control inputs received at the plant via wireless transmission, and $\mathbf{S} \succeq 0$ penalizes the terminal state, which is treated separately because no further control action is applied after the terminal time.
To stabilize the control process, the optimal control inputs $\{ \mathbf{u}_{n, \star} \}_{n=1}^{N-1}$ need to be derived to minimize the finite-horizon \ac{lqg} cost in \eqref{eq:control_cost_finite}.
Based on the channel-coherence assumption in Section~\ref{sect:system_model}, the beamforming matrices $\{\mathbf{W}_{\rm c}, \mathbf{W}_{\rm p}\}$ remain fixed over the $N$ time slots.
The finite-horizon control cost over $N$ time slots is characterized in the following theorem:
\begin{theorem} \label{theorem:finite_horizon}
    (Finite-horizon Control Cost) For a finite control process lasting $N$ steps, the optimal finite-horizon control cost is given by
    \begin{align}
        J_{N,\star} & =\mathbb{E} \left[ \mathbf{x}_{1}^{\mathsf{T}}\boldsymbol{\Theta}_1\mathbf{x}_{1}^{} \right] +\sum\nolimits_{n=1}^{N-1}{\mathrm{tr}\left\{ \boldsymbol{\Theta}_{n+1}\mathbf{\Sigma }_v \right\}}\notag \\
        &+\sum\nolimits_{n=1}^{N-1}{\mathrm{tr}\left\{ \left( \mathbf{D}+\mathbf{B}^{\mathsf{T}}\boldsymbol{\Theta}_{n+1}\mathbf{B} \right) \mathbf{\Sigma }_{e,n} \right\}} \notag \\
        &+ \sum\nolimits_{n=1}^{N-1}{\mathrm{tr}\left\{ \mathbf{K}_{n}^{\mathsf{T}}\mathbf{\Phi }_n\mathbf{K}_n\mathbf{C}_n \right\}}, \label{eq:min_control_cost}
    \end{align}
    where $\boldsymbol{\Theta}_n$ is the Riccati matrix for time slot $n$.   
    Additionally, the auxiliary matrix and the finite-horizon control gain matrix are respectively defined as follows:
    \begin{align}
        &\mathbf{\Phi }_{n}^{}\triangleq \bar{\mathbf{M}}_n^{\mathsf{T}}\left( \mathbf{D}+\mathbf{B}^{\mathsf{T}}\boldsymbol{\Theta}_{n+1}\mathbf{B} \right) \bar{\mathbf{M}}_n^{}, \notag \\
        &\mathbf{K}_n \triangleq \mathbf{\Phi }_{n}^{-1}\bar{\mathbf{M}}_n^{\mathsf{T}}\mathbf{B}^{\mathsf{T}}\boldsymbol{\Theta}_{n+1}\mathbf{A}. \notag
    \end{align}
    To evaluate $J_{N,\star}$, the Riccati matrix $\{\boldsymbol{\Theta}_n\}_{n=1}^{N-1}$ can be obtained from the following backward Riccati recursion:
    \begin{align}   
    \boldsymbol{\Theta}_n&=\mathbf{Q}+\mathbf{A}^{\mathsf{T}}\boldsymbol{\Theta}_{n+1}\mathbf{A}^{}-\mathbf{A}^{\mathsf{T}}\boldsymbol{\Theta}_{n+1}\mathbf{B}\bar{\mathbf{M}}_n^{}\mathbf{\Phi }_{n}^{-1}\bar{\mathbf{M}}_n^{\mathsf{T}}\mathbf{B}^{\mathsf{T}}\boldsymbol{\Theta}_{n+1}\mathbf{A}, \label{eq:riccati_matrix}
    \end{align}   
    where the terminal condition is specified by $\boldsymbol{\Theta}_N = \mathbf{S}$.
    To achieve the minimal control cost in \eqref{eq:min_control_cost}, the corresponding control inputs for a finite horizon $N$ are given by
    \begin{align}
        \mathbf{u}_{n,\star}=-\mathbf{K}_n \hat{\mathbf{x}}_n,\quad \textrm{for}~n=1,\ldots,N-1. \label{eq:optimal_control_input_finite}
    \end{align}
\end{theorem}
\begin{IEEEproof}
    See Appendix \ref{appendix:proof_finite_horizon}.
\end{IEEEproof}
To avoid horizon-dependent effects, we further consider the infinite-horizon average cost, which is obtained from the following equation:
\begin{align}
    J_{\infty}(\mathbf{W}_{\mathrm{c}},\mathbf{W}_{\mathrm{p}})=\lim_{N\rightarrow \infty} \frac{1}{N}J_{N, \star}(\mathbf{W}_{\mathrm{c}},\mathbf{W}_{\mathrm{p}}). \label{eq:infinite_lqr_cost}
\end{align}
We now consider a stationary control regime in which the second-order statistics of the control input and plant state converge to stationary matrices, i.e., $\boldsymbol{\Pi}_n\rightarrow\boldsymbol{\Pi}$ and $\boldsymbol{\Omega}_n\rightarrow\boldsymbol{\Omega}$, respectively.
Their stationary values are treated as fixed in the subsequent infinite-horizon control-cost characterization and beamforming design.
Thus, we have $\bar{\mathbf{M}}_n\rightarrow\bar{\mathbf{M}}$ and $\boldsymbol{\Sigma}_{e,n}\rightarrow\boldsymbol{\Sigma}_e$.
Under the standard stabilizability condition, the \ac{kf} estimation-error covariance also converges to its steady-state value, i.e., $\mathbf{C}_n\rightarrow\mathbf{C}$, which is supported by \cite{chan1984convergence}.
Building on the above, the infinite-horizon average cost in \eqref{eq:infinite_lqr_cost} can be characterized by the following theorem:
\begin{theorem}\label{theorem:infinite_horizon}
    (Infinite-horizon Average Control Cost) For an infinite-horizon control process, the optimal average control cost is given by
    \begin{align}
        J_{\infty}\left( \mathbf{W}_{\mathrm{c}},\mathbf{W}_{\mathrm{p}} \right) &=\mathrm{tr}\left\{ \boldsymbol{\Theta}\mathbf{\Sigma }_v \right\} +\mathrm{tr}\left\{ \left( \mathbf{D}+\mathbf{B}^{\mathsf{T}}\boldsymbol{\Theta}\mathbf{B} \right) \mathbf{\Sigma }_e \right\} \notag \\
        &+\mathrm{tr}\left\{ \mathbf{K}^{\mathsf{T}}\boldsymbol{\Phi} \mathbf{K} \mathbf{C}\right\} , \label{eq:control_cost_infinite_horizon}
    \end{align}
    where $\boldsymbol{\Theta}$ is the steady-state Riccati matrix.
    Additionally, the auxiliary matrix and the infinite-horizon control gain matrix are respectively defined as follows:
    \begin{align}
         \mathbf{\Phi }&\triangleq \bar{\mathbf{M}}^{\mathsf{T}}\left( \mathbf{D}+\mathbf{B}^{\mathsf{T}}\boldsymbol{\Theta}\mathbf{B} \right) \bar{\mathbf{M}},\notag \\
         \mathbf{K} &\triangleq \mathbf{\Phi }_{}^{-1}\bar{\mathbf{M}}^{\mathsf{T}}\mathbf{B}^{\mathsf{T}}\boldsymbol{\Theta}_{}\mathbf{A}. \notag 
    \end{align}
    To evaluate $J_\infty$, the steady-state Riccati matrix can be obtained from the infinite-horizon \ac{dare} given by 
    \begin{align}
        \boldsymbol{\Theta}&=\mathbf{Q}+\mathbf{A}^{\mathsf{T}}\boldsymbol{\Theta}\mathbf{A}-\mathbf{A}^{\mathsf{T}}\boldsymbol{\Theta}\mathbf{B}\bar{\mathbf{M}}\mathbf{\Phi }^{-1}\bar{\mathbf{M}}^{\mathsf{T}}\mathbf{B}^{\mathsf{T}}\boldsymbol{\Theta}\mathbf{A}. \label{eq:dare}
    \end{align}
    The corresponding control input for achieving the optimal infinite-horizon average control cost is given by 
    \begin{align}
        \mathbf{u}_{n, \star} = -\mathbf{K} \hat{\mathbf{x}}_n.
    \end{align}
\end{theorem}
\begin{IEEEproof}
    See Appendix \ref{appendix:proof_infinite_horizon}.
\end{IEEEproof}
Note that the infinite-horizon control cost in Theorem \ref{theorem:infinite_horizon} is defined under the stabilizing stationary regime.
Specifically, define the closed-loop state-transition matrix as $\mathbf{A}_{\mathrm{cl}}\triangleq \mathbf{A}-\mathbf{B}\bar{\mathbf{M}}\mathbf{K}$.
Under the standard conditions that $(\mathbf{A},\mathbf{B}\bar{\mathbf{M}})$ is stabilizable and $(\mathbf{A},\mathbf{Q}^{1/2})$ is detectable,  the \ac{dare} in \eqref{eq:dare} admits a stabilizing solution satisfying $\rho\left(\mathbf{A}_{\mathrm{cl}}\right)<1$, where $\rho(\cdot)$ denotes the spectral radius. \footnote{In particular, stabilizability means that there exists a feedback matrix $\mathbf{K}$ such that the closed-loop matrix $\mathbf{A}-\mathbf{B}\bar{\mathbf{M}}\mathbf{K}$ is stable, i.e., $\rho(\mathbf{A}-\mathbf{B}\bar{\mathbf{M}}\mathbf{K})<1$. 
Detectability means that any eigenvector $\mathbf{v}$ of $\mathbf{A}$ corresponding to an eigenvalue satisfying $|\lambda|\geq 1$ also satisfies $\mathbf{Q}^{1/2}\mathbf{v}\neq\mathbf{0}$.}
Accordingly, we have $\mathbf{A}_{\mathrm{cl}}^{n}\rightarrow\mathbf{0}_{L \times L}$ as $n\rightarrow\infty$. 
This result indicates that the plant-state covariance converges to a finite stationary value.
Together with the stationary disturbance and estimation-error covariances considered above, this ensures that the infinite-horizon average control cost is finite.
Regarding this theorem, we provide the following remarks:
\begin{remark} \label{remark:beamforming_on_control_cost}
    (How Beamforming Affects Control Cost) 
   \emph{The communication beamformer, i.e., $\mathbf{W}_{\rm c}$, governs the communication interference through the covariance matrix $\mathbf{R}_d$, while the control-input precoding matrix $\mathbf{W}_{\rm p}$ affects the control cost by manipulating the equivalent channel of the plant, i.e., $\mathbf{F}\mathbf{W}_{\rm p}$.
    In short, we can write $\bar{\mathbf{M}} = \bar{\mathbf{M}}(\mathbf{W}_{\rm c}, \mathbf{W}_{\rm p})$ and $\boldsymbol{\Sigma}_e = \boldsymbol{\Sigma}_e(\mathbf{W}_{\rm c}, \mathbf{W}_{\rm p})$ to highlight their dependence on the beamformers. }
\end{remark}
\begin{remark}
    (Solution to Riccati Fixed-Point Equation) 
    \emph{Based on Theorem \ref{theorem:infinite_horizon}, we show how to evaluate $J_{\infty}(\mathbf{W}_{\rm c}, \mathbf{W}_{\rm p})$ for fixed beamformers $\{\mathbf{W}_{\rm c}, \mathbf{W}_{\rm p}\}$ in what follows.
    Substituting the expression of $\boldsymbol{\Phi}$ into the steady-state Riccati equation yields the standard \ac{dare}.
    For fixed $\bar{\mathbf{M}}$ and $\boldsymbol{\Sigma}_e$, the \ac{dare} can be solved via the following steps: i) Initialize the Riccati matrix $\boldsymbol{\Theta}^{(0)}=\mathbf{Q}$ or any positive-definite matrix; ii) iteratively compute $\mathbf{\Phi }^{\left( i \right)}=\bar{\mathbf{M}}^{\mathsf{T}}\left( \mathbf{D}+\mathbf{B}^{\mathsf{T}}\boldsymbol{\Theta}^{\left( i \right)}\mathbf{B} \right) \bar{\mathbf{M}}^{}$ and $\boldsymbol{\Theta}^{\left( i+1 \right)}=\mathbf{Q}+\mathbf{A}^{\mathsf{T}}\boldsymbol{\Theta}^{\left( i \right)}\mathbf{A}-\mathbf{A}^{\mathsf{T}}\boldsymbol{\Theta}^{\left( i \right)}\mathbf{B}\bar{\mathbf{M}}\left( \mathbf{\Phi }^{\left( i \right)} \right) ^{-1}\bar{\mathbf{M}}^{\mathsf{T}}\mathbf{B}^{\mathsf{T}}\boldsymbol{\Theta}^{\left( i \right)}\mathbf{A}$ until the convergence condition is achieved, e.g., $\|\boldsymbol{\Theta}^{(i+1)} - \boldsymbol{\Theta}^{(i)}\|_F \le \eta$ where $\eta>0$ is a prescribed tolerance, and let $I$ denote the stopping iteration index; and iii) compute $J_{\infty}\left( \mathbf{W}_{\mathrm{c}},\mathbf{W}_{\mathrm{p}} \right)$ using $\boldsymbol{\Theta} = \boldsymbol{\Theta}^{(I)}$ and $\boldsymbol{\Phi} = \boldsymbol{\Phi}^{(I)}$.}
\end{remark}

\section{Beamforming Design for JCC Systems With Vector-Valued Control Inputs} \label{sect:beamforming_design_jcc_vector_valued}
As discussed in Remark~\ref{remark:beamforming_on_control_cost}, the communication beamforming matrix $\mathbf{W}_{\rm c}$ and the control-input precoding matrix $\mathbf{W}_{\rm p}$ jointly determine the communication performance and the infinite-horizon average control cost.
Since the infinite-horizon control cost is defined under the stabilizing stationary regime, we focus on optimizing the long-term control performance. 
Accordingly, the beamforming design problem is formulated as follows:
\begin{problem}\label{pb:vector_beamforming}
\begin{alignat}{2}
\underset{\mathbf{W}_{\rm c},\mathbf{W}_{\rm p}}{\rm{min}} &\quad J_{\infty}\left(\mathbf{W}_{\rm c},\mathbf{W}_{\rm p}\right) \label{obj:vector_bf}\\
{\rm s.t.}&\quad \gamma_k\left(\mathbf{W}_{\rm c},\mathbf{W}_{\rm p}\right) \geq \Gamma_{k}, \quad \forall k\in\mathcal{K}, \label{constraint:communication_qos}
\\&\quad \sum\nolimits_{k=1}^{K} \left\|\mathbf{w}_k\right\|_2^2 + \left\|\mathbf{W}_{\rm p}\right\|_F^2 \leq P_{\max}, \label{constraint:transmit_power}
\end{alignat}
\end{problem}
where $J_{\infty}\left(\mathbf{W}_{\rm c},\mathbf{W}_{\rm p}\right)$ denotes the infinite-horizon average control cost, constraint \eqref{constraint:communication_qos} ensures that the \ac{sinr} of the $k$-th \ac{cu} meets the prescribed target \ac{qos} requirement $\Gamma_k$, and constraint \eqref{constraint:transmit_power} limits the total transmit power to be no larger than $P_{\max}$.
We note that this formulation does not imply a higher priority for communication. 
The communication \ac{sinr} constraints specify the minimum \ac{qos} requirements, while the control performance is directly optimized by minimizing the long-term \ac{lqg} cost. Although a control-cost constraint can also be considered, its threshold is generally application-dependent and related to the specific plant dynamics and control task. Therefore, we adopt the current formulation in this work.
Problem \eqref{pb:vector_beamforming} is nonconvex and therefore challenging to solve directly. 
In particular, the objective \eqref{obj:vector_bf} is a nonlinear function of the beamforming matrices, owing to the \ac{lmmse} receiver and the associated steady-state control and estimation quantities.
In addition, constraint \eqref{constraint:communication_qos} is also nonconvex in its fractional form.
To this end, we first reformulate constraint \eqref{constraint:communication_qos} as an equivalent convex constraint and then address the nonconvexity in the objective function using an \ac{sca} approach.

Based on \cite{xu2026generalized}, constraint \eqref{constraint:communication_qos} can be equivalently reformulated as a second-order cone constraint.
The reformulation uses the invariance of the objective function and the constraints under a unit-modulus phase rotation of the beamforming vector $\mathbf{w}_k$.
In particular, replacing $\mathbf{w}_k$ with $\mathrm{e}^{\mathrm{j}\phi_k} \mathbf{w}_k$, where $\phi_k \in [0, 2 \pi)$, does not change the value of all relevant squared magnitudes, covariance matrices, and the transmit-power terms.
Without loss of optimality, we choose $\phi_k = -\angle (\mathbf{h}_k^{\mathsf{H}} \mathbf{w}_k)$.
As such, the following equations hold:
\begin{align}
    \Re\{\mathbf{h}_k^{\textsf{H}} \mathbf{w}_k \} \ge 0, \quad \Im\{\mathbf{h}_k^{\textsf{H}} \mathbf{w}_k \} =0. \label{eq:socp_re}
\end{align}
The \ac{qos} constraint can be equivalently written as follows:
\begin{align}
    \left| \mathbf{h}_{k}^{\mathsf{H}}\mathbf{w}_k \right|^2\ge \Gamma _k\left( \sum\nolimits_{j\ne k}^{}{\left| \mathbf{h}_{k}^{\mathsf{H}}\mathbf{w}_j \right|^2+\left\| \mathbf{h}_{k}^{\mathsf{H}}\mathbf{W}_{\mathrm{p}} \right\| _{2}^{2}}+\sigma _{\mathrm{c}}^{2} \right). \label{constraint:communication_qos_equi}
\end{align}
According to \eqref{eq:socp_re}, the phase condition $\left| \mathbf{h}_{k}^{\mathsf{H}}\mathbf{w}_k \right| = \Re\{\mathbf{h}_k^{\textsf{H}} \mathbf{w}_k \}$ holds. 
Thus, \eqref{constraint:communication_qos_equi} can be further recast as the following second-order cone constraint:
\begin{align}
    \left\| \left[ \begin{array}{c}
	\mathbf{h}_{k}^{\mathsf{H}}\mathbf{w}_1\\
	\vdots\\
	\mathbf{h}_{k}^{\mathsf{H}}\mathbf{w}_{k-1}\\
	\mathbf{h}_{k}^{\mathsf{H}}\mathbf{w}_{k+1}\\
	\vdots\\
	\mathbf{h}_{k}^{\mathsf{H}}\mathbf{w}_K\\
	\mathbf{W}_{\mathrm{p}}^{\mathsf{H}}\mathbf{h}_{k}^{}\\
	\sigma _{\mathrm{c}}\\
\end{array} \right] \right\| _2\le \frac{\Re \left\{ \mathbf{h}_{k}^{\mathsf{H}}\mathbf{w}_k \right\}}{\sqrt{\Gamma _k}}. \label{constraint:socp_qos}
\end{align}
The resulting inequality is a convex second-order cone constraint, since its left-hand side is the Euclidean norm of an affine mapping and its right-hand side is affine \cite{boyd2004convex}.
Analogously, the power constraint in \eqref{constraint:transmit_power} can also be reformulated as follows:
\begin{align}
    \left\| \left[ \begin{array}{c}
	\mathrm{vec}\left\{ \mathbf{W}_{\mathrm{c}} \right\}\\
	\mathrm{vec}\left\{ \mathbf{W}_{\mathrm{p}} \right\}\\
\end{array} \right] \right\| _2\le \sqrt{P_{\max}}. \label{constraint:socp_power}
\end{align}
At this point, all constraints have been converted into convex forms.

Subsequently, we address the nonconvex objective function. 
According to Theorem \ref{theorem:infinite_horizon}, the infinite-horizon average cost function depends nonlinearly on the beamforming matrices and therefore cannot be directly reformulated as an \ac{socp} objective. 
Therefore, we employ the \ac{sca} approach to construct a strongly convex local surrogate of the objective function at each iteration. 
In particular, define the real-valued optimization variable as follows:
\begin{align}
    \mathbf{p} =[\Re\{\mathrm{vec\{\mathbf{W}_{\rm c}\}}\}; \Im\{\mathrm{vec\{\mathbf{W}_{\rm c}\}}\}; \Re\{\mathrm{vec\{\mathbf{W}_{\rm p}\}}\}; \Im\{\mathrm{vec\{\mathbf{W}_{\rm p}\}}\}]. \notag
\end{align}
Let $f(\mathbf{p})$ be the shorthand for the objective function $J_{\infty}(\mathbf{W}_{\rm c}, \mathbf{W}_{\rm p})$.
At the $i$-th iteration, the current optimization variable, objective value, and gradient are denoted by $\mathbf{p}^{(i)}$, $f(\mathbf{p}^{(i)})$, and $\mathbf{g}^{(i)} = \nabla f(\mathbf{p}^{(i)})$, respectively.
Since $f(\cdot)$ involves the steady-state Riccati equation and \ac{kf} error-covariance recursion, its gradient does not admit a closed-form solution.
Therefore, this gradient is numerically computed using a central finite-difference approximation.

At iteration $i$, we construct a first-order approximation of $f(\cdot)$ around $\mathbf{p}^{(i)}$ with an augmented quadratic proximal term, which is given by
\begin{align}
    \check{f}_i\left( \mathbf{p} \right) =f(\mathbf{p}^{\left( i \right)})+( \mathbf{g}^{\left( i \right)}) _{}^{\mathsf{T}}( \mathbf{p}-\mathbf{p}^{\left( i \right)} ) +\frac{\tau _i}{2}\| \mathbf{p}-\mathbf{p}^{\left( i \right)} \| _{2}^{2}, \label{eq:surrogate_func}
\end{align}
where the quadratic proximal term penalizes large deviations from the current position $\mathbf{p}^{\left( i \right)}$ and weight factor $\tau_i > 0$ is used to maintain convexity. 
Specifically, the parameter $\tau_i$ controls the conservativeness of the update: A larger $\tau_i$ leads to a smaller step, and vice versa.
Hence, at iteration $i$, the original problem \eqref{pb:vector_beamforming} can be reformulated as follows:
\begin{problem}\label{pb:vector_beamforming_reformulate}
\begin{alignat}{2}
    \underset{\mathbf{p}}{\rm{min}} &\quad \check{f}_i\left( \mathbf{p} \right) \label{obj:vector_bf_i_th_iter}\\
    {\rm s.t.}&\quad \eqref{constraint:socp_qos}~\mathrm{and ~}\eqref{constraint:socp_power}. \notag 
\end{alignat}
\end{problem}
The candidate solution is obtained by minimizing the strongly convex surrogate function over the convex feasible set.
To streamline the following backtracking method, we denote the candidate solution obtained by solving problem \ref{pb:vector_beamforming_reformulate} as $\mathbf{p}^+$.

However, the convexity of the surrogate function alone does not guarantee a decrease in the original objective function.
In particular, a candidate solution $\mathbf{p}^+$ may cause $f(\mathbf{p}^+) > \check{f}_i\left( \mathbf{p}^+ \right)$, thereby pushing the minimization in the opposite direction.
Therefore, we use a backtracking procedure to select the weight factor $\tau_i$.
More specifically, let $\beta>1$ denote the multiplicative backtracking factor.
Starting from an initial value of $\tau_i$, we first solve the convex surrogate problem and examine whether the candidate solution decreases the objective value, i.e., $f(\mathbf{p}^+) \le \check{f}_i\left( \mathbf{p}^+ \right)$.
If not, we update $\tau_i$ by $\tau_i \leftarrow \beta \tau_i$.
The updated $\tau_i$ is then used to re-solve the surrogate problem, yielding a new candidate solution.
This procedure is repeated until $f(\mathbf{p}^+) \le \check{f}_i\left( \mathbf{p}^+ \right)$ is satisfied, thereby ensuring a decrease in the objective value.
The iterative procedure terminates when the relative change in the objective value satisfies the following condition:
\begin{align}
    \frac{\left|f\left(\mathbf{p}^{(i+1)}\right)-f\left(\mathbf{p}^{(i)}\right)\right|}{\max\left\{1,\left|f\left(\mathbf{p}^{(i)}\right)\right|\right\}}\leq \epsilon, \label{eq:convergence_condition}
\end{align}
where $\epsilon>0$ is a prescribed convergence tolerance.
The complete algorithm is presented in Algorithm \ref{alg:socp_sca}. 
\begin{algorithm}[t!]
    \small
    \caption{SOCP-Based SCA Algorithm}
    \label{alg:socp_sca}
    \begin{algorithmic}[1]
        \STATE{Initialize $i=0$, a feasible point $\mathbf{p}^{(0)}$, $\tau_0>0$, $\beta>1$, and $\epsilon>0$\;}
        \REPEAT
            \STATE{Evaluate $f(\mathbf{p}^{(i)})$ and compute the gradient $\mathbf{g}^{(i)}=\nabla f(\mathbf{p}^{(i)})$\;}
            \REPEAT
                \STATE{Construct the surrogate function $\check{f}_i(\mathbf{p})$ according to \eqref{eq:surrogate_func}\;}
                \STATE{Obtain the candidate point $\mathbf{p}^{+}$ by solving the convex problem in \eqref{pb:vector_beamforming_reformulate}\;}
                \IF{$f(\mathbf{p}^{+})> \check{f}_i(\mathbf{p}^{+})$}
                    \STATE{Update $\tau_i\leftarrow\beta\tau_i$\;}
                \ENDIF
            \UNTIL{$f(\mathbf{p}^{+})\leq\check{f}_i(\mathbf{p}^{+})$}
            \STATE{Set $\mathbf{p}^{(i+1)}=\mathbf{p}^{+}$, $\tau_{i+1}=\tau_i$, and $i=i+1$ \;}
        \UNTIL{the convergence condition in \eqref{eq:convergence_condition} is satisfied}
        \RETURN{$\mathbf{W}_{\rm c}^{\star}$ and $\mathbf{W}_{\rm p}^{\star}$ recovered from $\mathbf{p}^{(i)}$}
    \end{algorithmic}
\end{algorithm}
Let $I_{\rm SCA}$ and $I_{\rm BT}$ denote the numbers of \ac{sca} iterations and backtracking steps per iteration, respectively. 
Given that the real-valued optimization variable in $\mathbf{p}$ contains $2M_{\rm t}(K+L)$ entries, the complexity of solving each \ac{socp} problem via an interior-point method is given by $\mathcal{O}\left((M_{\rm t}(K+L))^{3.5}\log\left({1}/{\delta}\right)\right)$, where $\delta>0$ denotes the solver accuracy.
Thus, the overall computational complexity of Algorithm \ref{alg:socp_sca} is given by $\mathcal{O}\left(I_{\rm SCA}I_{\rm BT}(M_{\rm t}(K+L))^{3.5}\log\left({1}/{\delta}\right)\right)$.

\section{Beamforming Design for JCC Systems With Scalar-Valued Control Inputs} \label{sect:beamforming_design_jcc_scalar_valued}
To gain more insight into the trade-off in a \ac{jcc} system, we consider a simplified scenario in which the control inputs are scalars rather than vectors.
In practice, these scalar control inputs correspond to one-dimensional actuation commands.
In this section, we first specialize the vector control-input model in Section \ref{subsect:vector_control_input} to the scalar case. 
Based on this specialization, a single-\ac{cu}, single-plant scenario is considered to characterize the Pareto boundary in the \ac{jcc} system.
Note that we assume that the plant is equipped with a single antenna for scalar control-input reception in this section. 

\subsection{Specialization to Scalar Control Inputs}
Let $\mathbf{h}_{\rm c} \in \mathbb{C}^{M_{\rm t} \times 1}$ and $\mathbf{h}_{\rm p} \in \mathbb{C}^{M_{\rm t} \times 1}$ denote the communication and control channel vectors. 
The \ac{bs} adopts beamformers $\mathbf{w}_{\rm c} \in \mathbb{C}^{M_{\rm t} \times 1}$ and $\mathbf{w}_{\rm p} \in \mathbb{C}^{M_{\rm t} \times 1}$ to transmit communication symbols $s_n \in \mathbb{C}$ and the normalized control input $\tilde{u}_n \in \mathbb{R}$, respectively.
In particular, the normalized control input is obtained from $\tilde{u}_n=\pi_n ^{-1/2}u_n$, where $u_n$ is the original scalar control input and $\pi_n \triangleq \mathbb{E} [ u_{n}^{2}]$.
Under the assumptions that $\mathbb{E} [|s_n|^2 ] =1$ and $\mathbb{E} [\tilde{u}_n^2 ] =1$, the transmit signal is given by
\begin{align}
    \mathbf{s}_n=\mathbf{w}_{\mathrm{c}}s_n+\mathbf{w}_{\mathrm{p}}\tilde{u}_n.
\end{align}
Consequently, the received signals at the plant and the \ac{cu} are respectively given by
\begin{align}
    r_{{\rm p},n}&=\mathbf{h}_{\rm p}^{\textsf H}\mathbf{w}_{\rm p}\tilde u_n+ \mathbf{h}_{\rm p}^{\textsf H}\mathbf{w}_{\rm c}s_n+z_{{\rm p},n},
    \label{eq:scalar_plant_received_signal}\\
    y_{{\rm c},n}&=\mathbf{h}_{\rm c}^{\textsf H}\mathbf{w}_{\rm c}s_n+\mathbf{h}_{\rm c}^{\textsf H}\mathbf{w}_{\rm p}\tilde u_n+z_{{\rm c},n},\label{eq:scalar_cu_received_signal}
\end{align}
where $z_{{\rm p},n}\sim\mathcal{CN}(0,\sigma_{\rm p}^2)$ and $z_{{\rm c},n}\sim\mathcal{CN}(0,\sigma_{\rm c}^2)$ denote the additive Gaussian noise at the plant and the \ac{cu}, respectively.
Following the same steps used in the derivations for the vector control-input scenario, the \ac{sinr} for the \ac{cu} is given by
\begin{align}
    {\gamma}_{\rm c}(\mathbf{w}_{\rm c},\mathbf{w}_{\rm p}) &=
    \frac{|\mathbf{h}_{\rm c}^{\textsf H}\mathbf{w}_{\rm c}|^2}
    {|\mathbf{h}_{\rm c}^{\textsf H}\mathbf{w}_{\rm p}|^2+\sigma_{\rm c}^2}.
\end{align}
At the plant, we adopt the \ac{lmmse} estimator to extract $\tilde{u}_n$ from the received signal.
First, let $f_{\mathrm p}\triangleq \mathbf h_{\mathrm p}^{\mathsf H}\mathbf w_{\mathrm p}$ denote the effective control-channel coefficient from the \ac{bs} to the plant, and define the aggregate interference-plus-noise variance at the plant as
$\sigma_d^2\triangleq |\mathbf h_{\mathrm p}^{\mathsf H}\mathbf w_{\mathrm c} |^2 +\sigma_{\mathrm p}^2$.
Therefore, as in the vector case in \eqref{eq:plant_observation_vector}, the received signal at the plant is given by $r_{\mathrm{p},n}=f_{\mathrm{p}}\tilde{u}_n+d_n$, where $\mathbb{E}[|d_n|^2]=\sigma_d^2$.
Additionally, we stack the real and imaginary components of the effective channel and the received disturbance as follows:
\begin{align}
    &\breve{\mathbf{f}}_{\mathrm{p}}\triangleq \left[ \Re \left\{ f_{\mathrm{p}} \right\} ,\Im \left\{ f_{\mathrm{p}} \right\} \right] ^{\textsf{T}}, \quad \breve{\mathbf{d}}_n\triangleq \left[ \Re \left\{ d_n \right\} ,\Im \left\{ d_n \right\} \right] ^{\textsf{T}}, \notag  \\
    &\breve{\mathbf{r}}_n\triangleq\left[\Re\{r_{\mathrm{p},n}\},\Im\{r_{\mathrm{p},n}\}\right]^{\mathsf{T}}. \notag 
\end{align}
Thus, $r_{\mathrm{p},n}$ can be written as follows:
\begin{align}
    \breve{\mathbf{r}}_n=\breve{\mathbf{f}}_{\mathrm{p}}\tilde{u}_n+\breve{\mathbf{d}}_n. \label{eq:real_receive_control_input}
\end{align}
In this scalar case, the \ac{lmmse} combiner is given by 
\begin{align}
    \mathbf{g}^{\mathsf{T}}=\breve{\mathbf{f}}_{\mathrm{p}}^{\mathsf{T}}( \breve{\mathbf{f}}_{\mathrm{p}}^{}\breve{\mathbf{f}}_{\mathrm{p}}^{\mathsf{T}}+\breve{\mathbf{R}}_d) ^{-1}=\frac{\breve{\mathbf{f}}_{\mathrm{p}}^{\mathsf{T}}}{\left| f_{\mathrm{p}} \right|^2+\sigma _{d}^{2}/2}, \notag 
\end{align}
where $\breve{\mathbf{R}}_d\triangleq \mathbb{E} [ \breve{\mathbf{d}}_{n}^{}\breve{\mathbf{d}}_{n}^{\mathsf{T}} ] =\frac{\sigma _{d}^{2}}{2}\mathbf{I}_2$ denotes the covariance matrix of $\breve{\mathbf{d}}_n$.
Therefore, the estimated normalized control input is given by $\hat{\tilde{u}}_n=\mathbf{g}^{\mathsf{T}}\breve{\mathbf{r}}_n$.
Then, after denormalization using $\pi^{1/2}_n$, the estimated scalar control input is given by
\begin{align}
    \hat{u}_n=\pi ^{1/2}_n\hat{\tilde{u}}_n=\bar{m}_n{u}_n+e_n,
\end{align}
where the recovery coefficient and the actuation disturbance are respectively defined as follows:
\begin{align}
    \bar{m}_n \triangleq \mathbf{g}^{\mathsf{T}}\breve{\mathbf{f}}_{\mathrm{p}}^{}=\frac{\left| f_{\mathrm{p}} \right|^2}{\left| f_{\mathrm{p}} \right|^2+\sigma _{d}^{2}/2}, \quad e_n\triangleq \pi_n ^{1/2}\mathbf{g}^{\mathsf{T}}\breve{\mathbf{d}}_n. \label{eq:recovery_coeff_and_disturbance_variance}
\end{align}
Additionally, the variance of $e_n$ is given by 
\begin{align}
    \sigma _{e,n}^{2}=\mathbb{E} [e_{n}^{2}]=\pi _{n}^{}\mathbb{E} [ |\mathbf{g}^{\mathsf{T}}\breve{\mathbf{d}}_n|^2 ] =\frac{\left( \pi _n\sigma _{d}^{2}/2 \right) \left| f_{\mathrm{p}} \right|^2}{(\left| f_{\mathrm{p}} \right|^2+\sigma _{d}^{2}/2)^2}. \notag
\end{align}
Define $A \in \mathbb{R}$ and $B \in \mathbb{R}$ as the scalar counterparts of $\mathbf{A}$ and $\mathbf{B}$ in \eqref{eq:control_equation}, respectively.
The scalar state-evolution equation from time slot $n$ to time slot $n+1$ is characterized by
\begin{align}
    x_{n+1} = A x_{n} + B \bar{m}_n u_n +B e_n + v_n, \label{eq:state_evolution_scalar}
\end{align}
where $v_n \sim \mathcal{N}(0, \sigma_{\rm v}^2)$ denotes the additive control process noise.
Similar to the vector case, the \ac{bs} adopts the \ac{kf} to estimate the scalar plant state.
Hence, we also specialize the state-feedback mechanism to the scalar case. 
Note that, although the plant state is scalar-valued, the \ac{bs} utilizes a multi-antenna array to receive.
Let $\mathbf{h}_{\rm f}\in\mathbb{C}^{M_{\rm r}\times 1}$ denote the uplink Rayleigh fading channel from the plant to the \ac{bs}.
Letting $\omega_{n+1} \triangleq \mathbb{E}[x_{n+1}^2]$, the received signal at the \ac{bs} is given by 
\begin{align}
    \mathbf{c}_{n+1}={\sqrt{P_{\rm p}}} \omega_{n+1}^{-1/2} \mathbf{h}_{\rm f}x_{n+1}+\mathbf{n}_{n+1},
\end{align}
where $\mathbf{n}_{n+1}\sim\mathcal{CN}
(\boldsymbol{0}_{M_{\rm r}},\sigma_{\rm f}^2\mathbf{I}_{M_{\rm r}})$ denotes the additive Gaussian noise.
Here, beamforming at the plant is not implemented due to its single-antenna architecture.

For a general time slot, let $\hat{x}_{n}$ denote the posterior \ac{kf} estimate of $x_n$, and define the corresponding estimation error as $\epsilon_n \triangleq x_n - \hat{x}_n$.
Accordingly, the relationship between the true scalar plant state and its estimated counterpart is given by $x_n = \hat{x}_n + \epsilon_n$.
Let $\sigma_{\epsilon, n}^2 \triangleq \mathbb{E}[\epsilon_n^2]$ denote the corresponding \ac{kf} estimation-error variance.
This additive relationship is leveraged below to explicitly account for the additional control cost incurred due to imperfect plant-state estimation.

For the control-performance metric, the finite-horizon cost function of scalar \ac{lqg} control is defined as follows \cite{van1981certainty}:
\begin{align}
    J_N(\mathbf{w}_{\mathrm{c}},\mathbf{w}_{\mathrm{p}})=\mathbb{E} \left[ \sum\nolimits_{n=1}^{N-1}{\left( Qx_{n}^{2}+D\hat{u}_{n}^{2} \right)}+Sx_{N}^{2} \right], \label{eq:scalar_control_cost_finite}
\end{align}
where $Q\ge 0$ penalizes the deviation of the scalar plant state from the desired origin, $D > 0$ penalizes the actual scalar control input applied at the plant, and $S\ge 0$ penalizes the terminal state.
Again, to eliminate horizon-dependent effects, we define the infinite-horizon average cost as $J_{\infty}(\mathbf{w}_{\mathrm{c}},\mathbf{w}_{\mathrm{p}})=\lim_{N\rightarrow \infty} (1/N)J_N(\mathbf{w}_{\mathrm{c}},\mathbf{w}_{\mathrm{p}})$.
As in the vector case, we consider the stationary operating regime for the infinite-horizon analysis. 
Consequently, we have $\pi_n \rightarrow \pi$, $\bar{m}_n \rightarrow \bar{m}$, $\sigma_{e,n}^2 \rightarrow \sigma_e^2$, and $\sigma_{\epsilon, n}^2 \rightarrow \sigma_\epsilon^2$. 
Based on the above, the following corollary gives the scalar version of the infinite-horizon optimal control cost:
\begin{corollary} \label{corollary:control_input_scalar}
    (Infinite-horizon Average Control Cost) For the scalar control input, the optimal infinite-horizon average control cost is given by
    \begin{align}
        J_{\infty} (\mathbf{w}_{\rm c}, \mathbf{w}_{\rm p}) =  \theta \sigma_{v}^2 + (D+B^2\theta) \sigma_{e}^2 +\frac{A^2B^2\theta ^2}{D+B^2\theta}\sigma _{\epsilon}^{2}, \label{eq:infinite_horizon_avg_control_cost_scalar}
    \end{align}
    where the steady-state Riccati coefficient $\theta$ is given by
    \begin{align}
        \theta&=\frac{QB^2-D\left( 1-A^2 \right) +\sqrt{\Delta}}{2B^2}, \\
        \Delta &\triangleq \left( D\left( 1-A^2 \right) -QB^2 \right)^2 +4QDB^2.
    \end{align}
    The corresponding optimal scalar control input is given by 
    \begin{align}
        u_{n, \star} = -k\hat{x}_n,
    \end{align}
    where the scalar control gain is defined as $k\triangleq \frac{AB\theta}{\bar{m}\left( D+B^2\theta \right)}$.
\end{corollary}
\begin{IEEEproof}
    See Appendix \ref{appendix:proof_infinite_horizon_scalar}. 
\end{IEEEproof}
For the scalar case, the corresponding stability condition can be explicitly characterized.
According to the optimal control gain $k$, the closed-loop state-transition coefficient is given by $A_{\mathrm{cl}}\triangleq A-B\bar{m}k={AD}/{(D+B^2\theta)}$.
Therefore, the stabilizing condition reduces to $\left|A_{\mathrm{cl}}\right|=\left|\frac{AD}{D+B^2\theta}\right|<1$.
According to this condition, we have $A_{\mathrm{cl}}^{n}\rightarrow 0$ as $n\rightarrow\infty$, which guarantees that the plant-state variance converges to a finite stationary value. 
Therefore, the infinite-horizon average control cost in \eqref{eq:infinite_horizon_avg_control_cost_scalar} is finite.
Based on this corollary, we provide the following remark to explain the connection between the scalar and vector cases. 
\begin{remark}
    (Difference between the Scalar and Vector Cases)
    \emph{In the vector case, the steady-state Riccati matrix $\boldsymbol{\Theta}$ generally does not admit a closed-form solution and is typically obtained by numerically solving the corresponding \ac{dare}. 
    In contrast, in the scalar case, $\boldsymbol{\Theta}$ reduces to the scalar Riccati solution $\theta$, which admits a closed-form solution. 
    Similar to the vector case, the communication beamformer $\mathbf{w}_{\rm c}$ and the control beamformer $\mathbf{w}_{\rm p}$ affect the scalar control cost through the recovery coefficient $\bar{m}$ and the error variance $\sigma_e^2$.}
\end{remark}

\subsection{Pareto Boundary Characterization for the Scalar Case}
To characterize the trade-off between communication and control, we consider a Pareto optimization problem for this simplified scalar case.
Accordingly, for a given communication \ac{sinr} requirement
$\Gamma_{\rm c}$, the Pareto-boundary point can be obtained by
solving the following scalar beamforming problem:
\begin{problem}\label{pb:scalar_beamforming}
\begin{alignat}{2}
    \underset{\mathbf{w}_{\rm c},\mathbf{w}_{\rm p}}{\rm min} &\quad J_{\infty}\left(\mathbf{w}_{\rm c},\mathbf{w}_{\rm p}\right) \label{obj:scalar_bf} \\
{\rm s.t.} &\quad \gamma_{\rm c}\left(\mathbf{w}_{\rm c},\mathbf{w}_{\rm p} \right) \geq \Gamma_{\rm c}, \label{constraint:scalar_communication_qos} \\
&\quad \left\|\mathbf{w}_{\rm c}\right\|_2^2 + \left\|\mathbf{w}_{\rm p}\right\|_2^2 \leq P_{\max}. \label{constraint:scalar_transmit_power}
\end{alignat}
\end{problem}
To solve this problem, we first define the control-input \ac{sinr} as follows:
\begin{align}
    \rho _{\mathrm{p}}\left( \mathbf{w}_{\mathrm{c}},\mathbf{w}_{\mathrm{p}} \right) \triangleq \frac{2|\mathbf{h}_{\mathrm{p}}^{\mathsf{H}}\mathbf{w}_{\mathrm{p}}|^2}{|\mathbf{h}_{\mathrm{p}}^{\mathsf{H}}\mathbf{w}_{\mathrm{c}}|^2+\sigma _{\mathrm{p}}^{2}}, \label{eq:control_sinr}
\end{align}
where the factor of two follows from the complex-to-real mapping in \eqref{eq:real_valued_definitions}.
Based on this definition, the recovery coefficient $\bar{m}$ and the disturbance variance $\sigma_e^2$ in \eqref{eq:recovery_coeff_and_disturbance_variance} can be rewritten as follows:
\begin{align}
    \bar{m} =\frac{\rho_{\rm p}}{1+ \rho_{\rm p}},\quad \sigma_{e}^2 = \pi \frac{\rho_{\rm p}}{(1+\rho_{\rm p})^2}.  \notag
\end{align}
Using the above definitions, we present the following lemma to reformulate problem \eqref{pb:scalar_beamforming}.
\begin{lemma} \label{lemma:decreasing}
    (Monotonicity of $J_{\infty}(\mathbf{w}_{\rm c}, \mathbf{w}_{\rm p})$ w.r.t. $\rho_{\rm p}$) \emph{Whenever the effective closed-loop control coefficient $\bar{m}k$ is nonzero, the infinite-horizon average cost function is a strictly decreasing function with respect to the control \ac{sinr} $\rho_{\rm p}$, i.e., ${\mathrm{d} J_{\infty}}/{\mathrm{d} \rho_{\rm p}} <0$.}
\end{lemma}
\begin{IEEEproof}
    See Appendix \ref{appendix:decrease}.
\end{IEEEproof}
According to Lemma \ref{lemma:decreasing}, minimizing
$J_{\infty}$ is equivalent to maximizing the control \ac{sinr}
$\rho_{\rm p}$. 
Therefore, problem \eqref{pb:scalar_beamforming} can be reformulated as follows:
\begin{problem}\label{pb:scalar_pareto_boundary}
    \begin{alignat}{2}
    \underset{\mathbf{w}_{\rm c},\mathbf{w}_{\rm p}}{\rm max} &\quad \rho_{\rm p}\left(\mathbf{w}_{\rm c},\mathbf{w}_{\rm p}\right) \label{obj:scalar_control_sinr} \\
    {\rm s.t.} &\quad \gamma_{\rm c}\left(\mathbf{w}_{\rm c},\mathbf{w}_{\rm p}\right) \geq \Gamma_{\rm c}, \label{constraint:scalar_communication_qos_equivalent} \\
    &\quad \left\|\mathbf{w}_{\rm c}\right\|_2^2+\left\|\mathbf{w}_{\rm p}\right\|_2^2\leq P_{\max}. \label{constraint:scalar_transmit_power_equivalent}
    \end{alignat}
\end{problem}
By varying $\Gamma_{\rm c}$ over its feasible range, the Pareto boundary is obtained.
For each given $\Gamma_{\rm c}$, we employ a two-level procedure, where the outer loop searches for the maximum achievable control \ac{sinr} and the inner loop solves an \ac{socp} feasibility problem to determine the corresponding beamformers, i.e., $\{\mathbf{w}_{\rm c},\mathbf{w}_{\rm p} \}$.
More specifically, let $\varrho^{(i)} > 0$ be an auxiliary variable in the $i$-th outer-loop iteration, for which the constraint $\rho_{\rm p}\left(\mathbf{w}_{\rm c},\mathbf{w}_{\rm p}\right) \ge \varrho^{(i)}$ is imposed.
For a fixed $\varrho^{(i)}$, the objective function in \eqref{obj:scalar_control_sinr} can be represented in a second-order cone form according to Section~\ref{sect:beamforming_design_jcc_vector_valued}.
Additionally, constraint \eqref{constraint:scalar_communication_qos_equivalent} can also be represented in a second-order cone form, while constraint \eqref{constraint:scalar_transmit_power_equivalent} is already convex.
Consequently, for a fixed $\varrho^{(i)}$, the resulting inner problem is a convex problem and can be efficiently solved using standard convex optimization tools.

Outside the inner loop, the outer loop aims to determine the largest feasible value of $\varrho^{(i)}$.
Moreover, since feasibility is monotonic w.r.t. $\varrho^{(i)}$, the maximum feasible control \ac{sinr} can be efficiently determined via a bisection search.
Before performing the bisection search, we check the feasibility of the communication constraint, which requires $P_{\max}\geq\frac{\Gamma_{\rm c}\sigma_{\rm c}^{2}} {\|\mathbf{h}_{\rm c}\|_2^{2}}$.
If this inequality does not hold, the problem is infeasible, since the communication \ac{qos} requirement cannot be met even though all available power has been allocated to the communication beamformer and aligned with the direction of $\mathbf{h}_{\rm c}$. 
Once the problem's feasibility is checked, the lower and upper bounds of the bisection search interval can be obtained.
Since the control-input \ac{sinr} is non-negative, the lower bound is given by $\varrho_{\rm lb} = 0$.
Given that $|\mathbf{h}_{\rm p}^{\mathsf{H}}\mathbf{w}_{\rm p}|^2 \le P_{\max} \|\mathbf{h}_{\rm p}\|_2^2$ and $|\mathbf{h}_{\rm p}^{\mathsf{H}}\mathbf{w}_{\rm c}|^2 + \sigma_{\rm p}^2 \ge \sigma_{\rm p}^2$, the upper bound of the bisection search interval is given by $\varrho_{\rm ub} = {2P_{\max}\|\mathbf{h}_{\rm p}\|_2^2}/{\sigma_{\rm p}^2}$.
The overall algorithm is presented in Algorithm \ref{alg:scalar_socp_bisection}.
\begin{algorithm}[t!]
    \small
    \caption{SOCP-Based Bisection Algorithm}
    \label{alg:scalar_socp_bisection}
    \begin{algorithmic}[1]
        \STATE{Initialize $i=0$, $\varrho_{\rm lb}=0$, $\varrho_{\rm ub}$,
        and $\epsilon>0$\;}
        \REPEAT
            \STATE{Set
            $\varrho^{(i)}=(\varrho_{\rm lb}+\varrho_{\rm ub})/2$\;}
            \STATE{Solve the \ac{socp} feasibility problem \eqref{pb:scalar_pareto_boundary} with
            $\varrho=\varrho^{(i)}$\;}
            \IF{the problem is feasible}
                \STATE{Set $\varrho_{\rm lb}=\varrho^{(i)}$ and store the
                obtained beamformers\;}
            \ELSE
                \STATE{Set $\varrho_{\rm ub}=\varrho^{(i)}$\;}
            \ENDIF
            \STATE{Set $i=i+1$\;}
        \UNTIL{$\varrho_{\rm ub}-\varrho_{\rm lb}\leq\epsilon$}
        \RETURN{$\rho_{\rm p}^{\star}=\varrho_{\rm lb}$,
        $\mathbf{w}_{\rm c}^{\star}$, and
        $\mathbf{w}_{\rm p}^{\star}$}
    \end{algorithmic}
\end{algorithm}
Let $I_{\rm BIS}$ denote the number of bisection iterations.
Its complexity is given by $I_{\rm BIS}=\mathcal{O}\left(
\log_2((\varrho_{\rm ub}-\varrho_{\rm lb})/\epsilon)\right)$.
As the real-valued optimization variables contain $4M_{\rm t}$ entries, the complexity of solving each \ac{socp} via an interior-point method is $\mathcal{O}\left(M_{\rm t}^{3.5}\log(1/\delta)\right)$, where $\delta>0$ denotes the solver accuracy. 
Hence, the overall computational complexity of Algorithm \ref{alg:scalar_socp_bisection} is given by $\mathcal{O}\left(I_{\rm BIS}M_{\rm t}^{3.5}\log\left({1}/{\delta}\right)\right)$.

\begin{table}[t]
\centering
\caption{Simulation Parameters.}
\label{tab:sim_params}
\small
\renewcommand{\arraystretch}{0.95}
\setlength{\tabcolsep}{4pt}
\begin{tabular}{l c}
\hline
\textbf{Parameter} & \textbf{Value} \\
\hline
BS antennas $(M_{\rm t},M_{\rm r})$ & $(8,6)$ \\
Plant antennas/state dimension $(M_{\rm p},L)$ & $(2,2)$ \\
Number of CUs $K$ & $2$ \\
Power budgets $(P_{\max},P_{\rm p})$ & $(20,10)$ \\
Noise variances $(\sigma_{\rm c}^2,\sigma_{\rm p}^2,\sigma_{\rm f}^2)$
& $(1,1,1)$ \\
Process-noise covariance $\boldsymbol{\Sigma}_v$ & $0.02\mathbf I_L$ \\
LQG weights $(\mathbf Q,\mathbf D)$
& $(\mathbf I_L,0.2\mathbf I_L)$ \\
Plant input matrix $\mathbf B$ & $\mathbf I_L$ \\
Default CU SINR threshold $\Gamma$ & $3~\mathrm{dB}$ \\
MC runs / control-cost slots / tracking slots & $200/500/100$ \\
\hline
\end{tabular}
\end{table}

\section{Numerical Results} \label{sect:numerical_results}
The parameters in Table~\ref{tab:sim_params} are used for all simulations unless otherwise specified. 
In addition to these fixed parameters, all downlink and uplink channel coefficients are independently generated according to $\mathcal{CN}(0,1)$. 
For the vector case, the diagonal entries of $\mathbf{A}$ are uniformly spaced between $1$ and $1.5$, whereas its off-diagonal entries are independently drawn uniformly from $[0,0.1]$.
Additionally, $\mathbf{B}$ is set to $\mathbf{B}=\mathbf{I}_L$.
These settings not only account for open-loop dynamics but also ensure that the plant is controllable and can theoretically be stabilized by the feedback controller.
Several figure-specific settings are adopted according to the purpose of each experiment. 
For the scalar control-cost and state-tracking experiments, $(A,B)$ is set to $(1.25,1)$ and $(1.12,1)$, respectively. 
These settings allow us to isolate the effects of stochastic plant-state evolution and analyze the impact of channel-induced perturbations.
In the convergence test, $\mathbf{A}=1.2\mathbf{I}_L$, $M_{\rm t}=10$, $K=6$, and $\Gamma=5~\mathrm{dB}$ are used to provide a common operating point for the scalar and vector cases.
To validate the effectiveness of the proposed method, we introduce a benchmark referred to as ``Joint-ZF.''
In this benchmark, the communication and control beamforming directions are jointly designed to eliminate the interference between the two functionalities; then, the minimum power required to satisfy the communication \ac{sinr} constraints is allocated to the \acp{cu}, while the remaining power is used for control-input transmission.

Moreover, as indicated by \eqref{eq:control_cost_finite}, the \ac{lqg} cost penalizes both the plant-state deviation from the desired origin and the applied control effort. 
When the plant is stabilized around the desired origin, the absolute control cost is typically small, i.e., near zero. 
Hence, to better illustrate the control-performance loss incurred by supporting communication, we adopt the relative control-cost degradation as a normalized performance metric in the numerical results.
In particular, for each channel realization, this metric is defined as follows:
\begin{align}
    \Delta _J\triangleq \frac{J_{\infty}-J_{\infty}^{\mathrm{CO}}}{J_{\infty}^{\mathrm{CO}}}\times 100\%,
\end{align}
where $J_{\infty}$ is shorthand for $J_{\infty}(\mathbf{W}_{\rm c}, \mathbf{W}_{\rm p})$ and $J_{\infty}(\mathbf{w}_{\rm c}, \mathbf{w}_{\rm p})$ for the vector and scalar cases, respectively, and $J_{\infty}^{\rm CO}$ denotes the minimum infinite-horizon control cost achieved in the control-only case, i.e., when the precoder is optimized solely for control-cost minimization without imposing the communication \ac{sinr} constraint.
Note that $\Delta_J\ge0$ is introduced only for presenting and interpreting the numerical results, whereas the proposed beamforming designs are obtained by directly minimizing the original objectives of problems \eqref{pb:vector_beamforming} and \eqref{pb:scalar_beamforming}. 
Since $J_{\infty}^{\rm CO}>0$ is constant for a given channel realization, minimizing $\Delta_J$ is equivalent to minimizing $J_{\infty}$.
Hence, this normalization preserves the optimality of the solutions to the formulated problems. 
In fact, $\Delta_J$ measures the percentage increase in the long-term \ac{lqg} cost caused by satisfying the communication requirement relative to the best achievable control-only performance.
Accordingly, a smaller $\Delta_J$ indicates a more favorable communication-control trade-off.

\begin{figure}[t]
    \centering
    \subfloat[Vector Case.
    \label{fig:cost_convergence_vect}]{ 
        \includegraphics[height=0.5\linewidth]
        {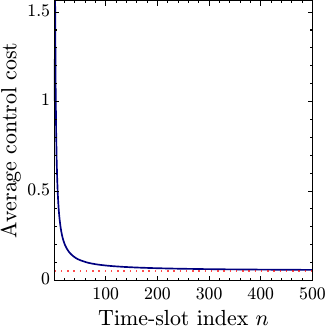}
    }
    \subfloat[Scalar Case.
    \label{fig:cost_convergence_sca}]{ 
        \includegraphics[height=0.5\linewidth]
        {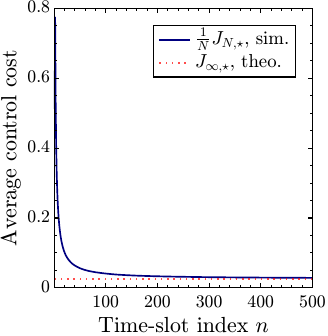}
    }
    \caption{Illustration of the convergence of the control cost for (a) the vector case and (b) the scalar case.}
    \label{fig:convergence}
\end{figure}
Fig.~\ref{fig:convergence} compares the empirical running-average \ac{lqg} costs with the corresponding theoretical infinite-horizon costs for both vector- and scalar-valued control inputs.
In particular, for a horizon of $N$ time slots, the empirical cost is averaged over the first $N$ time slots. 
As shown in Figs. \ref{fig:cost_convergence_vect} and \ref{fig:cost_convergence_sca}, the running-average control costs converge to their respective theoretical infinite-horizon values as $N$ increases, thereby validating Theorem~\ref{theorem:infinite_horizon} and Corollary~\ref{corollary:control_input_scalar}. 
The results also illustrate the transition of the closed-loop system from its initial transient response to steady-state operation under the designed control inputs.

\begin{figure}[t]
    \centering
    \subfloat[Vector Case.
    \label{fig:kf_convergence_vect}]{
        \includegraphics[height=0.5\linewidth] 
        {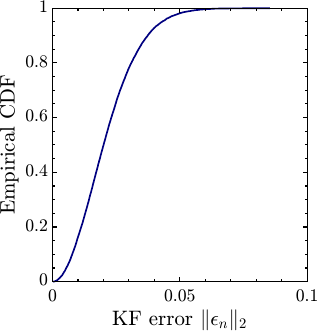}
    }
    \subfloat[Scalar Case.
    \label{fig:kf_convergence_sca}]{
        \includegraphics[height=0.5\linewidth] 
        {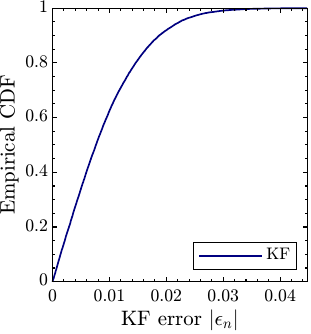}
    }
    \caption{Illustration of the empirical CDF of the \ac{kf} tracking process for (a) the vector case and (b) the scalar case. } \label{fig:kf_tracking_empirical_error}
\end{figure}
Fig.~\ref{fig:kf_tracking_empirical_error} illustrates the \ac{kf}-based state-tracking performance for both the vector and scalar cases. 
Specifically, the empirical \acp{cdf} are computed using the tracking errors over the last 100 time slots of 200 Monte Carlo runs, since the steady-state tracking performance of the \ac{kf} is of primary interest. 
Additionally, the tracking error is quantified using $\|\boldsymbol{\epsilon}_n\|_2$ for the vector case and $|\epsilon_n|$ for the scalar case.
As observed from the figure, most error realizations are below $0.05$ in the vector case and $0.04$ in the scalar case, thereby validating the effectiveness of the \ac{kf}-based tracking approach.
Taken together, Figs. \ref{fig:convergence} and \ref{fig:kf_tracking_empirical_error} validate the established wireless closed-loop control process, comprising downlink control-input transmission and uplink plant-state tracking.

\begin{figure}[t]
    \centering
    \subfloat[Vector Case.]{
        \includegraphics[height=0.5\linewidth] 
        {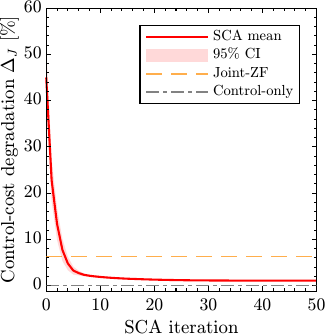}
    }
    \subfloat[Scalar Case.]{
        \includegraphics[height=0.5\linewidth] 
        {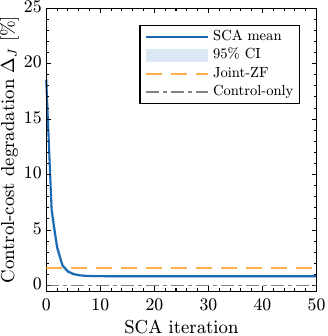}
    }
    \caption{Illustration of the convergence behavior for (a) the vector case and (b) the scalar case.} \label{fig:convergence_sca}
\end{figure}
Fig.~\ref{fig:convergence_sca} illustrates the convergence behavior of the proposed algorithms, where the 95~\% \ac{ci} quantifies the uncertainty in the mean estimated from the valid Monte Carlo trials.
In both cases, the control-cost degradation decreases monotonically and converges to a value lower than that of the ``Joint-ZF'' benchmark, verifying the effectiveness of the proposed algorithms.
The slower convergence in the vector case can be attributed to the stronger coupling among multiple control streams and communication beamformers.
The gap between the converged $\Delta_J$ and that of the \ac{zf} benchmark indicates the benefits of jointly designing the precoders instead of heuristically nulling interference.
Building on the above, we separately investigate the communication-control trade-off in the vector and scalar cases in Figs. \ref{fig:control_cost_degradation_vector} and \ref{fig:control_cost_degradation_scalar}, respectively.
The figures are presented separately for the following reasons: i) In the scalar case, a single \ac{cu} and one plant are considered, resulting in the Pareto boundary being defined by the optimal solution; and ii) in the vector case, the proposed \ac{sca}-based approach yields only a suboptimal solution but is capable of addressing more complex scenarios, i.e., higher-dimensional control inputs and multiple \acp{cu}.

\begin{figure}[t]
\centering
\includegraphics[width=0.85\linewidth]{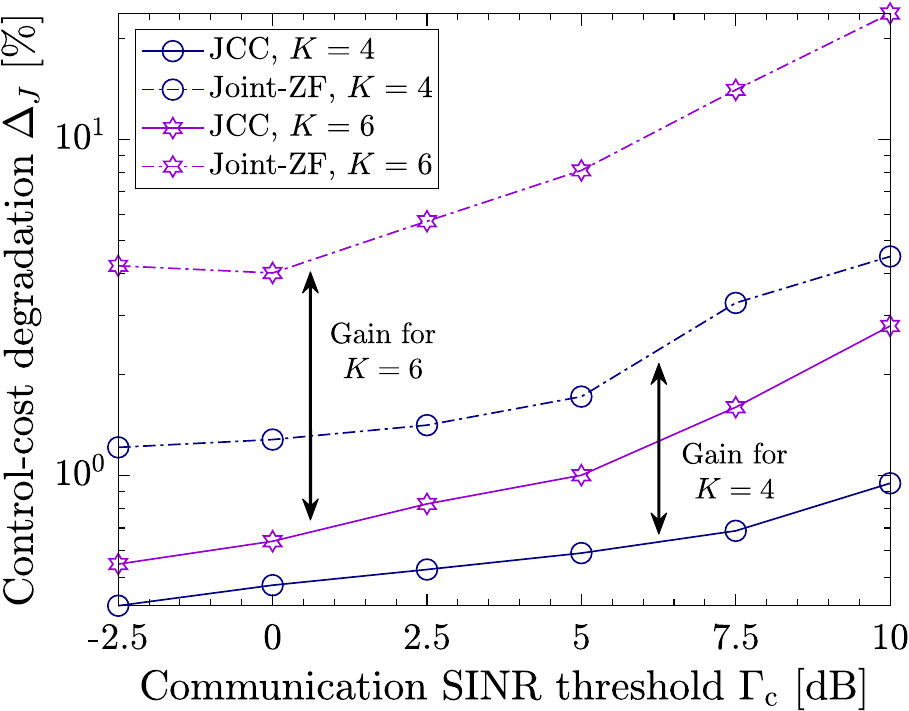} 
\caption{Illustration of the control-cost degradation versus the communication \ac{sinr} threshold for $M_{\rm t}=10$.}
\label{fig:control_cost_degradation_vector}
\end{figure}
Fig.~\ref{fig:control_cost_degradation_vector} illustrates the communication-control trade-off in the vector case for $M_{\rm t}=10$. 
In particular, $\Delta_J$, which represents the relative control-performance degradation, increases with both the communication \ac{sinr} threshold and the number of \acp{cu}.
This is because more transmit power and spatial \ac{dofs} must be allocated to communication to satisfy the minimum \ac{qos} requirements.
Moreover, the proposed \ac{jcc} design consistently outperforms the ``Joint-ZF'' benchmark, reducing $\Delta_J$ by 79.0~\% for $K=4$ at $7.5$~dB and by 87.6~\% for $K=6$ at $5$~dB.
These results indicate: i) When sufficient spatial \ac{dofs} are available, the ``Joint-ZF'' benchmark can achieve a performance comparable to that of the proposed algorithm; and ii) when the spatial \ac{dofs} are limited, completely nulling the interference becomes overly conservative, particularly when more \acp{cu} compete for the available spatial resources. 
In this case, the proposed algorithm achieves better performance by balancing inter-user and inter-function interference rather than completely eliminating it using \ac{zf}.

\begin{figure}[t]
\centering
\includegraphics[width=0.85\linewidth] 
{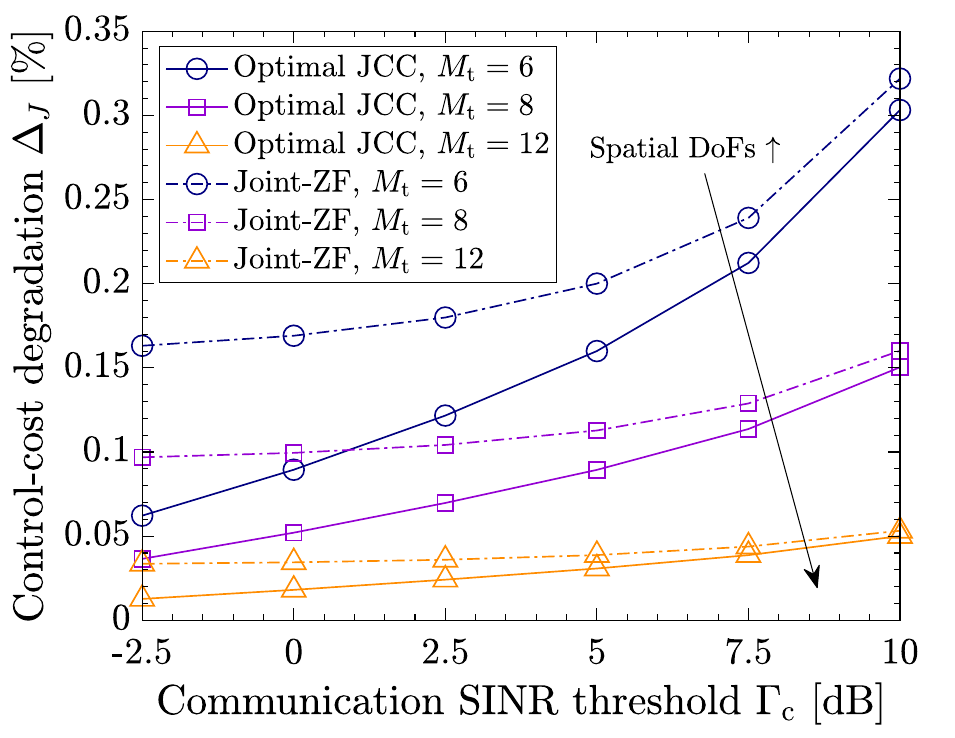}
\caption{Illustration of the optimal communication-control trade-off for the single-CU, single-plant scalar case for different $M_{\rm t}$.}
\label{fig:control_cost_degradation_scalar}
\end{figure}
Fig.~\ref{fig:control_cost_degradation_scalar} characterizes the optimal Pareto trade-off between communication and control in the scalar case.
In this case, the \ac{jcc} problem admits a globally optimal solution. 
As shown in the figure, the control-performance degradation measured by $\Delta_J$ increases with the communication \ac{sinr} threshold due to the competition between the communication and control functionalities.
Moreover, increasing the number of transmit antennas, denoted by $M_{\rm t}$, improves the control performance by providing additional array gain and spatial \ac{dofs}. 
It can also be observed that the proposed \ac{jcc} design consistently outperforms the ``Joint-ZF'' benchmark, as the former achieves the globally optimal solution in this simplified case. 
As shown in Fig.~\ref{fig:control_cost_degradation_scalar}, more stringent communication \ac{sinr} requirements necessitate complete interference suppression, which narrows the performance gap between the optimal solution and the benchmark.
In contrast, reducing the number of transmit antennas limits the available spatial \ac{dofs}, thereby widening the performance gap.

\section{Conclusions} \label{sect:conclusions}
This paper investigated a \ac{mimo} \ac{jcc} system, in which a multi-antenna \ac{bs} simultaneously served multiple \acp{cu} and controlled a physical plant.
Since the communication signals and control inputs were transmitted over shared spatial resources, both inter-user and inter-function interference needed to be addressed.
A wireless closed-loop control framework was considered, including downlink control-input recovery using an \ac{lmmse} receiver at the plant, as well as uplink state reporting and \ac{kf}-based tracking.
Within this framework, the finite-horizon control cost was first derived and then generalized to its infinite-horizon counterpart under the stationary and stabilizing conditions.
The derivations accounted for control-process noise, wireless actuation impairments, and state-estimation errors, thereby coupling beamforming with the long-term control process.
For vector-valued control inputs, a \ac{jcc} beamforming problem was formulated to minimize the infinite-horizon cost subject to communication \ac{qos} and transmit-power constraints. 
The resulting nonconvex problem was addressed using an \ac{socp}-based \ac{sca} algorithm with backtracking. 
For the scalar control-input case, a closed-form expression for the infinite-horizon control cost was first derived, and the communication-control Pareto boundary was optimally characterized for a one-\ac{cu}-one-plant scenario using an \ac{socp}-based bisection algorithm.
Simulation results verified the derived \ac{lqg} control cost, the \ac{kf}-based state tracking, and the convergence of the proposed algorithms. 
Furthermore, they demonstrated that properly balancing inter-function interference is more effective than completely nulling it via \ac{zf}, particularly when the spatial resources are limited.

\appendices
\section{Implementation of Kalman Filter}\label{appendix:implementation_of_kf}
Since the KF method is well documented, we simply provide a sketch of the implementation in this appendix.
The purpose of the \ac{kf} approach is to enable tracking of the true plant state $\mathbf{x}_n$ at the \ac{bs}.
According to the signal model in Section \ref{sect:uplink_modeling}, the state-transition model and the observation model can be expressed as follows:
\begin{align}
    \begin{cases}
	\mathbf{x}_n=\mathbf{Ax}_{n-1}+\mathbf{B}\bar{\mathbf{M}}_{n-1}\mathbf{u}_{n-1}+\mathbf{Be}_{n-1}+\mathbf{v}_{n-1},\\
	\mathbf{c}_n=\bar{\mathbf{H}}_n\mathbf{x}_n+\mathbf{n}_n, \\
\end{cases}
\end{align}
where the control input $\mathbf{u}_{n-1}$ is generated by the \ac{bs} and is therefore available to the \ac{bs}, whereas the instantaneous realizations of the recovery error $\mathbf{e}_{n-1}$ and the control-process noise $\mathbf{v}_{n-1}$ are generally unknown to the \ac{bs}.
Based on the statistical independence assumptions in Section~\ref{sect:system_model}, $\mathbf{e}_{n-1}$ and $\mathbf{v}_{n-1}$ are zero-mean and mutually independent, and their covariance matrices are given by $\boldsymbol{\Sigma}_{e, n-1}$ and $\boldsymbol{\Sigma}_v$, respectively.
Thus, the effective process noise in the \ac{kf} module is defined as $\mathbf{q}_{n-1} \triangleq \mathbf{Be}_{n-1}+\mathbf{v}_{n-1}$, which is zero-mean with a covariance of $\boldsymbol{\Sigma}_{q,n-1} = \mathbb{E}[\mathbf{q}_{n-1}\mathbf{q}_{n-1}^{\mathsf{T}}] = \mathbf{B}\boldsymbol{\Sigma}_{e, n-1}\mathbf{B}^{\mathsf{T}} + \boldsymbol{\Sigma}_v$.
Since the plant state is real-valued, we define the following variables:
\begin{align}
    \breve{\mathbf{c}}_n\triangleq \left[ \begin{array}{c}
	\Re \left\{ \mathbf{c}_{n}^{} \right\}\\
	\Im \left\{ \mathbf{c}_n \right\}\\
\end{array} \right] \in \mathbb{R} ^{2M_{\mathrm{r}}\times 1},\breve{\mathbf{H}}_n\triangleq \left[ \begin{array}{c}
	\Re \{\bar{\mathbf{H}}_n\}\\
	\Im \{\bar{\mathbf{H}}_n\}\\
\end{array} \right] \in \mathbb{R} ^{2M_{\mathrm{r}}\times L}. \notag
\end{align}
By using these definitions, the observation model can be rewritten as $\breve{\mathbf{c}}_n = \breve{\mathbf{H}}_n \mathbf{x}_n + \breve{\mathbf{n}}_n$, where $\breve{\mathbf{n}}_n \sim \mathcal{N}(\boldsymbol{0}_{2 M_{\mathrm{r}}}, \mathbf{R})$ and $\mathbf{R} \triangleq ({\sigma_{f}^2}/{2})\mathbf{I}_{2 M_{\mathrm{r}}}$.
Based on the above model, the \ac{kf} recursion at the BS consists of the following prediction and correction steps.
\begin{itemize}
    \item \textbf{Prior Prediction:} Based on the state-evolution model, the prior state estimate can be computed as $\hat{\mathbf{x}}_{n\mid n-1} = \mathbf{A}\hat{\mathbf{x}}_{n-1 \mid n-1} + \mathbf{B} \bar{\mathbf{M}}_{n-1} \mathbf{u}_{n-1}$, with the unknown $\mathbf{e}_{n-1}$ and $\mathbf{v}_{n-1}$ treated as process noise.
    Correspondingly, the prior error covariance matrix is updated as $\mathbf{C}_{n \mid n-1} = \mathbf{A}\mathbf{C}_{n-1 \mid n-1}\mathbf{A}^{\mathsf{T}} + \boldsymbol{\Sigma}_{q, n-1}$, where $\mathbf{C}_{n \mid n-1}$ and $\mathbf{C}_{n-1 \mid n-1}$ denote the prior and posterior \ac{kf} error covariance matrices, respectively. 

    \item \textbf{Posterior Update:} After receiving the uplink state report from the plant, the \ac{bs} computes the innovation vector as $\boldsymbol{\xi }_n=\breve{\mathbf{c}}_n-\breve{{\mathbf{H}}}_n\hat{\mathbf{x}}_{n\mid n-1}$ and the innovation covariance matrix $\mathbf{S}_n=\breve{\mathbf{H}}_n^{}\mathbf{C}_{n\mid n-1}\breve{\mathbf{H}}_n^{\mathsf{T}}+\mathbf{R}$.
    As such, the Kalman gain can be computed as $\mathbf{L}_n=\mathbf{C}_{n\mid n-1}\breve{\mathbf{H}}_n^{\mathsf{T}}\mathbf{S}_{n}^{-1}$.
    Based on the Kalman gain, the posterior state estimate is updated as $\hat{\mathbf{x}}_{n\mid n}=\hat{\mathbf{x}}_{n\mid n-1}+\mathbf{L}_n\boldsymbol{\xi }_n$, while the posterior error covariance is updated as $\mathbf{C}_{n\mid n}=( \mathbf{I}_L-\mathbf{L}_n\breve{\mathbf{H}}_n) \mathbf{C}_{n\mid n-1}$.
\end{itemize}
This completes the \ac{kf} implementation. 

\section{Proof of Theorem \ref{theorem:finite_horizon}} \label{appendix:proof_finite_horizon}
To prove this theorem, we first define the value function with the estimated plant state $\hat{\mathbf{x}}_n$ as follows:
\begin{align}
    &V_n\left( \hat{\mathbf{x}}_{n} \right) \triangleq \notag \\
    &\min _{\left\{ \mathbf{u}_i \right\} _{i=n}^{N-1}}\mathbb{E} \left[ \sum\nolimits_{i=n}^{N-1}{\left( \mathbf{x}_{i}^{\mathsf{T}}\mathbf{Qx}_{i}^{}+\hat{\mathbf{u}}_{i}^{\mathsf{T}}\mathbf{D}\hat{\mathbf{u}}_i \right) +\mathbf{x}_{N}^{\mathsf{T}}\mathbf{Sx}_{N}^{}} \mid \hat{\mathbf{x}}_n \right], \notag
\end{align}
which characterizes the minimum expected cumulative cost from time slot $n$ to the terminal time slot, given the current \ac{kf} state estimate $\hat{\mathbf{x}}_n$.
Recall that the true plant state and its \ac{kf} estimate satisfy $\mathbf{x}_n = \hat{\mathbf{x}}_n + \boldsymbol{\epsilon}_n$, where $\mathbb{E}[\boldsymbol{\epsilon}_n \mid \hat{\mathbf{x}}_n] = \boldsymbol{0}_L$ and $\mathbf{C}_n \triangleq \mathbb{E}[\boldsymbol{\epsilon}_n^{}\boldsymbol{\epsilon}_n^{\mathsf{T}}]$.
Since no further control input is applied after the terminal state, the terminal value function is given by $V_N\left( \hat{\mathbf{x}}_N \right) =\mathbb{E} \left[ \mathbf{x}_{N}^{\mathsf{T}}\mathbf{Sx}_{N}^{}\mid \hat{\mathbf{x}}_{N}^{} \right] =\hat{\mathbf{x}}_{N}^{\mathsf{T}}\mathbf{S}\hat{\mathbf{x}}_{N}^{}+\mathrm{tr}\left\{ \mathbf{SC}_N \right\}$.
To derive the optimal control policy, we apply Bellman's principle in a backward fashion. Specifically, for the estimated state at time slot $N-1$, Bellman's principle gives
\begin{align}
    V_{N-1}\left( \hat{\mathbf{x}}_{N-1} \right) &=\min_{\mathbf{u}_{N-1}} \mathbb{E} \left[ \mathbf{x}_{N-1}^{\mathsf{T}}\mathbf{Qx}_{N-1}+\hat{\mathbf{u}}_{N-1}^{\mathsf{T}}\mathbf{D}\hat{\mathbf{u}}_{N-1}\right. \notag \\
    &\left.+V_N\left( \hat{\mathbf{x}}_N \right) \mid \hat{\mathbf{x}}_{N-1} \right] . \label{eq:bellman_N_1}
\end{align}
Recall that $\hat{\mathbf{u}}_{N-1}=\bar{\mathbf{M}}_{N-1}\mathbf{u}_{N-1}+\mathbf{e}_{N-1}$ and $\mathbf{x}_N=\mathbf{Ax}_{N-1}+\mathbf{B}\bar{\mathbf{M}}_{N-1}\mathbf{u}_{N-1}+\mathbf{Be}_{N-1}+\mathbf{v}_{N-1}$.
Using the relationship $\mathbf{x}_{N-1} \triangleq \hat{\mathbf{x}}_{N-1} + \boldsymbol{\epsilon}_{N-1}$, the terminal state can be further written as follows:
\begin{align}
    \mathbf{x}_N=\boldsymbol{\mu }_N+\mathbf{A}\boldsymbol{\epsilon }_{N-1}+\mathbf{Be}_{N-1}+\mathbf{v}_{N-1},
\end{align}
where $\boldsymbol{\mu }_N\triangleq \mathbf{A} \hat{\mathbf{x}}_{N-1}+\mathbf{B}\bar{\mathbf{M}}_{N-1}\mathbf{u}_{N-1}$.
Moreover, using the fact that $\boldsymbol{\epsilon}_{N-1}$, $\mathbf{e}_{N-1}$, and $\mathbf{v}_{N-1}$ are zero-mean and mutually independent, the three expected cost terms in \eqref{eq:bellman_N_1} can be respectively computed as follows:
\begin{align}
    \mathbb{E} \left[ \mathbf{x}_{N-1}^{\mathsf{T}}\mathbf{Qx}_{N-1}^{}\mid \hat{\mathbf{x}}_{N-1}^{} \right] &=\hat{\mathbf{x}}_{N-1}^{\mathsf{T}}\mathbf{Q}\hat{\mathbf{x}}_{N-1}^{}+\mathrm{tr}\left\{ \mathbf{QC}_{N-1} \right\}, \notag \\
    \mathbb{E} \left[ \hat{\mathbf{u}}_{N-1}^{\mathsf{T}}\mathbf{D} \hat{\mathbf{u}}_{N-1}^{}\mid \hat{\mathbf{x}}_{N-1}^{} \right] &=\mathbf{u}_{N-1}^{\mathsf{T}}\bar{\mathbf{M}}_{N-1}^{\mathsf{T}}\mathbf{D}\bar{\mathbf{M}}_{N-1}\mathbf{u}_{N-1}^{} \notag \\
    &+\mathrm{tr}\left\{ \mathbf{D}\boldsymbol{\Sigma }_{e,N-1} \right\}, \notag\\
    \mathbb{E} \left[ \mathbf{x}_{N}^{\mathsf{T}}\mathbf{Sx}_{N}^{}\mid \hat{\mathbf{x}}_{N-1}^{} \right] &=\boldsymbol{\mu }_{N}^{\mathsf{T}}\mathbf{S}\boldsymbol{\mu }_{N}^{}+\mathrm{tr}\left\{ \mathbf{A}^{\mathsf{T}}\mathbf{SA}^{}\mathbf{C}_{N-1} \right\} \notag \\
    &+\mathrm{tr}\left\{ \mathbf{SB}\boldsymbol{\Sigma }_{e,N-1}\mathbf{B}^{\mathsf{T}} \right\} +\mathrm{tr}\left\{ \mathbf{S}\boldsymbol{\Sigma }_v \right\}. \notag
\end{align}
Substituting the above three expressions back into the Bellman equation yields the following equation:
\begin{align}
    &V_{N-1}\left( \hat{\mathbf{x}}_{N-1} \right) =\notag \\
    &\quad \min_{\mathbf{u}_{N-1}} \left\{ \hat{\mathbf{x}}_{N-1}^{\mathsf{T}}\mathbf{Q}\hat{\mathbf{x}}_{N-1}^{}+\mathbf{u}_{N-1}^{\mathsf{T}}\bar{\mathbf{M}}_{N-1}^{\mathsf{T}}\mathbf{D}\bar{\mathbf{M}}_{N-1}\mathbf{u}_{N-1} \right. \notag \\
    &\quad\left.+\boldsymbol{\mu }_{N}^{\mathsf{T}}\mathbf{S}\boldsymbol{\mu }_{N}^{}+\mathrm{tr}\left\{ \left( \mathbf{Q}+\mathbf{A}^{\mathsf{T}}\mathbf{SA}^{} \right) \mathbf{C}_{N-1} \right\} \right. \notag \\
    &\quad\left.+\mathrm{tr}\left\{ \left( \mathbf{D}+\mathbf{B}^{\mathsf{T}}\mathbf{SB} \right) \boldsymbol{\Sigma }_{e,N-1} \right\} +\mathrm{tr}\left\{ \mathbf{S}\boldsymbol{\Sigma }_v \right\} \right\}. \label{eq:bellman_N_2}
\end{align}
To derive the optimal control input, we further expand 
$\boldsymbol{\mu}_{N}^{\mathsf{T}}\mathbf{S}\boldsymbol{\mu}_{N}$ in the above equation and collect all the terms containing $\mathbf{u}_{N-1}$.
Therefore, \eqref{eq:bellman_N_2} can be rewritten as follows:
\begin{align}
    &V_{N-1}\left( \hat{\mathbf{x}}_{N-1} \right) =\min_{\mathbf{u}_{N-1}} \left\{ \mathbf{u}_{N-1}^{\mathsf{T}}\boldsymbol{\Phi }_{N-1}\mathbf{u}_{N-1} +\mathrm{tr}\left\{ \mathbf{S}\boldsymbol{\Sigma }_v \right\}\right.\notag \\
    &\left. +2\hat{\mathbf{x}}_{N-1}^{\mathsf{T}}\mathbf{A}^{\mathsf{T}}\mathbf{SB}\bar{\mathbf{M}}_{N-1}\mathbf{u}_{N-1}+\hat{\mathbf{x}}_{N-1}^{\mathsf{T}}\left( \mathbf{Q}+\mathbf{A}^{\mathsf{T}}\mathbf{SA} \right) \hat{\mathbf{x}}_{N-1} \right. \notag \\
    &\left.+\mathrm{tr}\left\{ \left( \mathbf{Q}+\mathbf{A}^{\mathsf{T}}\mathbf{SA} \right) \mathbf{C}_{N-1} \right\} +\mathrm{tr}\left\{ \left( \mathbf{D}+\mathbf{B}^{\mathsf{T}}\mathbf{SB} \right) \boldsymbol{\Sigma }_{e,N-1} \right\}  \right\},
    \label{eq:bellman_N_3}
\end{align}
where the auxiliary matrix $\boldsymbol{\Phi}_{N-1}$ is defined as
\begin{align}
    \boldsymbol{\Phi}_{N-1}\triangleq\bar{\mathbf{M}}_{N-1}^{\mathsf{T}}\left(\mathbf{D}+\mathbf{B}^{\mathsf{T}}\mathbf{S}\mathbf{B}\right)\bar{\mathbf{M}}_{N-1}.
\end{align}
Leveraging the quadratic expression in \eqref{eq:bellman_N_3}, the optimal control input $\mathbf{u}_{N-1, \star}$ can be derived by taking the gradient with respect to $\mathbf{u}_{N-1}$ and setting it equal to zero.
Hence, the optimal control input is given by 
\begin{align}
    \mathbf{u}_{N-1,\star}=-\boldsymbol{\Phi }_{N-1}^{-1}\bar{\mathbf{M}}_{N-1}^{\mathsf{T}}\mathbf{B}^{\mathsf{T}}\mathbf{SA}\hat{\mathbf{x}}_{N-1}= -\mathbf{K}_{N-1}\hat{\mathbf{x}}_{N-1}, \notag 
\end{align}
where the control gain matrix is defined as $\mathbf{K}_{N-1}\triangleq \boldsymbol{\Phi }_{N-1}^{-1}\bar{\mathbf{M}}_{N-1}^{\mathsf{T}}\mathbf{B}^{\mathsf{T}}\mathbf{SA}$.
By substituting $\mathbf{u}_{N-1, \star}$ back into the value function in \eqref{eq:bellman_N_2}, we have 
\begin{align}
    &V_{N-1}\left( \hat{\mathbf{x}}_{N-1} \right) =\notag \\
    &\qquad \qquad \hat{\mathbf{x}}_{N-1}^{\mathsf{T}}\boldsymbol{\Theta }_{N-1}\hat{\mathbf{x}}_{N-1}^{}+\mathrm{tr}\left\{ \boldsymbol{\Theta }_{N-1}\mathbf{C}_{N-1} \right\} +c_{N-1}, \notag 
\end{align}
where the Riccati matrix and the constant term are respectively defined as follows:
\begin{align}
    \boldsymbol{\Theta }_{N-1}&\triangleq \mathbf{Q}+\mathbf{A}^{\mathsf{T}}\mathbf{SA}-\mathbf{A}^{\mathsf{T}}\mathbf{SB}\bar{\mathbf{M}}_{N-1}\boldsymbol{\Phi }_{N-1}^{-1}\bar{\mathbf{M}}_{N-1}^{\mathsf{T}}\mathbf{B}^{\mathsf{T}}\mathbf{SA}, \notag \\
    c_{N-1}&\triangleq \mathrm{tr}\left\{ \mathbf{S}\boldsymbol{\Sigma }_v \right\} +\mathrm{tr}\left\{ \left( \mathbf{D}+\mathbf{B}^{\mathsf{T}}\mathbf{SB} \right) \boldsymbol{\Sigma }_{e, N-1} \right\} \notag \\
    &+\mathrm{tr}\left\{ \mathbf{K}_{N-1}^{\mathsf{T}}\boldsymbol{\Phi }_{N-1}^{}\mathbf{K}_{N-1}^{}\mathbf{C}_{N-1} \right\}. \notag 
\end{align}
Note that the last term in $c_{N-1}$ quantifies the additional control cost due to the \ac{kf} state-estimation error.
Therefore, $V_{N-1}\left( \hat{\mathbf{x}}_{N-1} \right)$ has the same quadratic form as the terminal value function $V_{N}\left( \hat{\mathbf{x}}_{N} \right)$ with an additional constant.
Following the above derivations, we can apply Bellman's recursion backward from time slot $N$.
For example, for time slot $N-2$, the optimal value function is obtained by solving 
\begin{align}
    V_{N-2}\left( \hat{\mathbf{x}}_{N-2} \right) &=\min_{\mathbf{u}_{N-2}} \mathbb{E} \left[ \mathbf{x}_{N-2}^{\mathsf{T}}\mathbf{Qx}_{N-2}^{}+\hat{\mathbf{u}}_{N-2}^{\mathsf{T}}\mathbf{D}\hat{\mathbf{u}}_{N-2}^{}\right. \notag \\
    &\left.+V_{N-1}\left( \hat{\mathbf{x}}_{N-1}^{} \right) \mid \hat{\mathbf{x}}_{N-2}^{} \right]. \notag 
\end{align}
Repeating the same argument backward from time slot $N-1$ to time slot $1$, we obtain 
\begin{align}
    \mathbf{u}_{n,\star} = - \mathbf{K}_n \hat{\mathbf{x}}_n, 
\end{align}
where the control gain and the auxiliary matrix are respectively given by 
\begin{align}
    \mathbf{K}_n &\triangleq \mathbf{\Phi }_{n}^{-1}\bar{\mathbf{M}}_{n}^{\mathsf{T}}\mathbf{B}^{\mathsf{T}}\boldsymbol{\Theta}_{n+1}\mathbf{A}, \\
    \boldsymbol{\Phi }_n&\triangleq \bar{\mathbf{M}}_n^{\mathsf{T}}\left( \mathbf{D}+\mathbf{B}^{\mathsf{T}}\boldsymbol{\Theta }_{n+1}\mathbf{B}^{} \right) \bar{\mathbf{M}}_n^{}.
\end{align}
The value function at time slot $n$ is given by 
\begin{align}
    V_n(\hat{\mathbf{x}}_n) = \hat{\mathbf{x}}_n^{\textsf{T}}\boldsymbol{\Theta}_n\hat{\mathbf{x}}_n^{} + \mathrm{tr}\{\boldsymbol{\Theta}_n \mathbf{C}_n \} + c_n,
\end{align}
where the following definitions are needed:
\begin{align}
    \boldsymbol{\Theta }_n&\triangleq \mathbf{Q}+\mathbf{A}^{\mathsf{T}}\boldsymbol{\Theta }_{n+1}\mathbf{A} \notag \\
    &-\mathbf{A}^{\mathsf{T}}\boldsymbol{\Theta }_{n+1}\mathbf{B}\bar{\mathbf{M}}_n^{}\boldsymbol{\Phi }_{n}^{-1}\bar{\mathbf{M}}_n^{\mathsf{T}}\mathbf{B}^{\mathsf{T}}\boldsymbol{\Theta }_{n+1}\mathbf{A},
    \\
    c_n &=c_{n+1}+l_n, \\
   l_n&\triangleq \mathrm{tr}\left\{ \boldsymbol{\Theta }_{n+1}\boldsymbol{\Sigma }_v \right\} +\mathrm{tr}\left\{ \left( \mathbf{D}+\mathbf{B}^{\mathsf{T}}\boldsymbol{\Theta }_{n+1}\mathbf{B} \right) \boldsymbol{\Sigma }_{e,n} \right\} \notag \\
   &+\mathrm{tr}\left\{ \mathbf{K}_{n}^{\mathsf{T}}\boldsymbol{\Phi }_{n}^{}\mathbf{K}_{n}^{}\mathbf{C}_n \right\}.
\end{align}
Since $c_N=0$, we have $c_1 = \sum\nolimits_{n=1}^{N-1}l_n$.
Moreover, given that $V_1(\hat{\mathbf{x}}_1)$ is the value function associated with the initial \ac{kf} state estimate, the optimal finite-horizon control cost can be expressed as follows:
\begin{align}
    J_{N,\star}=\mathbb{E} \left[ V_1\left( \hat{\mathbf{x}}_1 \right) \right] =\mathbb{E} \left[ \hat{\mathbf{x}}_{1}^{\mathsf{T}}\boldsymbol{\Theta }_1\hat{\mathbf{x}}_{1}^{} \right] +\mathrm{tr}\left\{ \boldsymbol{\Theta }_1\mathbf{C}_1 \right\} +\sum\nolimits_{n=1}^{N-1}{l_n}. \notag
\end{align}
By leveraging $\mathbb{E} \left[ \mathbf{x}_{1}^{\mathsf{T}}\boldsymbol{\Theta }_1\mathbf{x}_{1}^{} \right] =\mathbb{E} \left[ \hat{\mathbf{x}}_{1}^{\mathsf{T}}\boldsymbol{\Theta }_1\hat{\mathbf{x}}_{1}^{} \right] +\mathrm{tr}\left\{ \boldsymbol{\Theta }_1\mathbf{C}_1 \right\} $, we finally obtain 
\begin{align}
    J_{N,\star}&=\mathbb{E} \left[ \mathbf{x}_{1}^{\mathsf{T}}\boldsymbol{\Theta }_1\mathbf{x}_{1}^{} \right] +\sum\nolimits_{n=1}^{N-1}{\mathrm{tr}\left\{ \boldsymbol{\Theta }_{n+1}\boldsymbol{\Sigma }_v \right\}}\notag \\
    &+\sum\nolimits_{n=1}^{N-1}{\mathrm{tr}\left\{ \left( \mathbf{D}+\mathbf{B}^{\mathsf{T}}\boldsymbol{\Theta }_{n+1}\mathbf{B} \right) \boldsymbol{\Sigma }_{e,n} \right\}}\notag \\
    &+\sum\nolimits_{n=1}^{N-1}{\mathrm{tr}\left\{ \mathbf{K}_{n}^{\mathsf{T}}\boldsymbol{\Phi }_{n}^{}\mathbf{K}_{n}^{}\mathbf{C}_n \right\}}.
\end{align}
This completes the proof of Theorem \ref{theorem:finite_horizon}.

\section{Proof of Theorem \ref{theorem:infinite_horizon}} \label{appendix:proof_infinite_horizon}
In this appendix, we investigate the generalization from the finite-horizon expression of $J_{N, \star}$ to the infinite-horizon counterpart. 
Based on the results in Theorem \ref{theorem:finite_horizon}, we define the optimal finite-horizon average cost as follows: 
\begin{align}
    \frac{1}{N}J_{N,\star}&=\frac{1}{N}\mathbb{E} \left[ \mathbf{x}_{1}^{\mathsf{T}}\boldsymbol{\Theta}_1\mathbf{x}_{1}^{} \right] +\frac{1}{N}\sum\nolimits_{n=1}^{N-1}{\mathrm{tr}\left\{ \boldsymbol{\Theta}_{n+1}\mathbf{\Sigma }_v \right\}} \notag \\
    & +\frac{1}{N}\sum\nolimits_{n=1}^{N-1}{\mathrm{tr}\left\{ \left( \mathbf{D}+\mathbf{B}^{\mathsf{T}}\boldsymbol{\Theta}_{n+1}\mathbf{B} \right) \mathbf{\Sigma }_{e,n} \right\}} \notag \\
    &+\frac{1}{N}\sum\nolimits_{n=1}^{N-1}{\mathrm{tr}\left\{ \mathbf{K}_{n}^{\mathsf{T}}\mathbf{\Phi }_n\mathbf{K}_n\mathbf{C}_n \right\}},
\end{align}
where $\{\boldsymbol{\Theta}_n\}_{n=1}^N$ denotes the Riccati matrices defined recursively in \eqref{eq:riccati_matrix}.
To explicitly indicate the dependence on the finite horizon $N$, we include the horizon length as a superscript.
Accordingly, the Riccati matrix at time slot $n$ is denoted by $\boldsymbol{\Theta}_n^{(N)}$.
Under the standard stabilizing conditions for the infinite-horizon quadratic control problem, the finite-horizon Riccati recursion converges to the steady-state Riccati matrix, which is given by
\begin{align}
    \boldsymbol{\Theta}_{n}^{(N)} \rightarrow \boldsymbol{\Theta},\quad N-n \rightarrow \infty, \label{eq:steady_status}
\end{align}
where $\boldsymbol{\Theta}$ denotes the stabilizing solution to the infinite-horizon \ac{dare} given by
\begin{align}
\boldsymbol{\Theta}&=\mathbf{Q}+\mathbf{A}^{\mathsf{T}}\boldsymbol{\Theta}\mathbf{A}-\mathbf{A}^{\mathsf{T}}\boldsymbol{\Theta}\mathbf{B}\bar{\mathbf{M}}\mathbf{\Phi }^{-1}\bar{\mathbf{M}}^{\mathsf{T}}\mathbf{B}^{\mathsf{T}}\boldsymbol{\Theta}\mathbf{A}. 
\end{align}
Under the stationary conditions $\boldsymbol{\Phi}_n \rightarrow \boldsymbol{\Phi}$ and $\mathbf{K}_n \rightarrow \mathbf{K}$, their respective expressions are given by
\begin{align}
    \mathbf{\Phi }&= \bar{\mathbf{M}}^{\mathsf{T}}\left( \mathbf{D}+\mathbf{B}^{\mathsf{T}}\boldsymbol{\Theta}\mathbf{B} \right) \bar{\mathbf{M}}^{}, \notag \\
    \mathbf{K}&=\mathbf{\Phi }_{}^{-1}\bar{\mathbf{M}}^{\mathsf{T}}\mathbf{B}^{\mathsf{T}}\boldsymbol{\Theta}\mathbf{A}. \notag 
\end{align}
Moreover, under the stationary operating regime detailed in Theorem \ref{theorem:infinite_horizon}, the posterior estimation-error covariance converges as $\mathbf{C}_n \rightarrow \mathbf{C}$.

Given that $\boldsymbol{\Theta}_1$ remains bounded and $\mathbb{E}[\mathbf{x}_1^{\mathsf{T}} \mathbf{x}_1]$ is finite, we have $\lim _{N\rightarrow \infty}\frac{1}{N}\mathbb{E} \left[ \mathbf{x}_{1}^{\mathsf{T}}\boldsymbol{\Theta}_1\mathbf{x}_{1}^{} \right] =0$.
Therefore, by letting $N\rightarrow \infty$, the infinite-horizon average cost is given by
\begin{align}
    J_{\infty}=\mathrm{tr}\left\{ \boldsymbol{\Theta} \mathbf{\Sigma }_v \right\} +\mathrm{tr}\left\{ \left( \mathbf{D}+\mathbf{B}^{\mathsf{T}}\boldsymbol{\Theta} \mathbf{B} \right) \mathbf{\Sigma }_e \right\} + \mathrm{tr}\left\{ \mathbf{K}_{}^{\mathsf{T}}\boldsymbol{\Phi}\mathbf{K} \mathbf{C}\right\}. \notag 
\end{align}
Furthermore, by taking the steady-state limit of the finite-horizon optimal control policy, the corresponding infinite-horizon control input is given by $\mathbf{u}_{n,\star}=-\mathbf{K}\hat{\mathbf{x}}_n$.
This completes the proof of Theorem \ref{theorem:infinite_horizon}.

\section{Proof of Corollary \ref{corollary:control_input_scalar}} \label{appendix:proof_infinite_horizon_scalar}
This corollary can be proved by scalarizing the derivations in Appendices \ref{appendix:proof_finite_horizon} and \ref{appendix:proof_infinite_horizon}.
In terms of notation, we replace $\mathbf{A}$, $\mathbf{B}$, $\mathbf{Q}$, $\mathbf{D}$, $\boldsymbol{\Theta}$, $\bar{\mathbf{M}}$, $\boldsymbol{\Sigma}_e$, and $\boldsymbol{\Sigma}_v$ with their scalar counterparts $A$, $B$, $Q$, $D$, $\theta$, $\bar{m}$, $\sigma_e^2$, and $\sigma_v^2$, respectively.

Thus, the auxiliary matrix $\boldsymbol{\Phi}$ reduces to $\phi = \bar{m}^2 (D + B^2\theta)$.
Under the condition that $\bar{m} \neq 0$ and the stabilizing solution satisfies $D + B^2\theta > 0$, the scalar auxiliary term $\phi$ is invertible.
Substituting $\phi$ into the Riccati equation yields
\begin{align}
    \theta&=Q+A^2\theta-A\theta B\bar{m}\phi ^{-1}\bar{m}B\theta A \notag \\
    &=Q+A^2 \theta -\frac{A^2B^2\theta^2}{D+B^2\theta}. \label{eq:p_scalar}
\end{align}
The matrix traces in \eqref{eq:control_cost_infinite_horizon}, i.e., $\mathrm{tr}\left\{ \boldsymbol{\Theta}\boldsymbol{\Sigma }_v \right\} $, $\mathrm{tr}\left\{ \left( \mathbf{D}+\mathbf{B}^{\mathsf{T}}\boldsymbol{\Theta}\mathbf{B} \right) \boldsymbol{\Sigma }_e \right\}$, and $\mathrm{tr}\left\{ \mathbf{K}_{}^{\mathsf{T}}\boldsymbol{\Phi}\mathbf{K}\mathbf{C} \right\}$, reduce to $\theta \sigma_v^2$, $(D+B^2\theta)\sigma_{e}^2$, and $k^2 \phi \sigma_{\epsilon}^2$, respectively.
Here, the scalar control gain is given by
\begin{align}
    k = \phi^{-1} \bar{m} B \theta A = \frac{AB\theta}{\bar{m}(D+B^2\theta)}. \notag
\end{align}
Therefore, the infinite-horizon average control cost for the scalar control input can be expressed as follows:
\begin{align}
    J_{\infty} (\mathbf{w}_{\rm c}, \mathbf{w}_{\rm p}) = \theta \sigma_{v}^2 + (D+B^2\theta) \sigma_{e}^2 + \frac{A^2B^2\theta^2}{D+B^2\theta} \sigma_{\epsilon}^2. 
\end{align}
In contrast to the vector case, the scalar Riccati equation admits a closed-form solution.
In particular, rearranging \eqref{eq:p_scalar} gives $B^2 \theta^2 + (D(1 - A^2) - QB^2)\theta - QD = 0$.
Solving this quadratic equation and selecting the nonnegative root yields
\begin{align}
    \theta =\frac{QB^2-D\left( 1-A^2 \right) +\sqrt{\Delta}}{2B^2}, \label{eq:appendix_theta}
\end{align}
where $\Delta \triangleq \left( D\left( 1-A^2 \right) -QB^2 \right)^2+4QDB^2$.
Since $\Delta \ge 0$, the scalar Riccati equation, i.e., \eqref{eq:p_scalar}, admits real-valued solutions. 
This completes the proof.

\section{Proof of Lemma \ref{lemma:decreasing}} \label{appendix:decrease}
This proof proceeds in two steps: i) Separating $\rho_{\rm p}$ from the infinite-horizon average control cost, and ii) taking the first-order derivative of the resulting expression.
For notational simplicity, define $ I_1 \triangleq \bar{m}k =\frac{AB\theta}{D+B^2\theta}$, which is independent of the beamforming vectors. 
Given that $\frac{A^2B^2\theta^2}{D+B^2\theta}=(D+B^2\theta)I_1^2$, the infinite-horizon average control cost in
\eqref{eq:infinite_horizon_avg_control_cost_scalar} can be rewritten as follows:
\begin{align}
    J_{\infty}=\theta \sigma _{v}^{2}+(D+B^2\theta )\left( \sigma _{e}^{2}+I_{1}^{2}\sigma _{\epsilon}^{2} \right), \label{eq:cost_I1}
\end{align}
where $\sigma_e^2$ and $\sigma_\epsilon^2$ are parameterized by $\rho_{\rm p}$.
Therefore, we aim to reveal this dependence in an analytical form. 

Recall that $u_{n,\star}=-k\hat{x}_n$ and $\pi=\mathbb{E}[u_{n,\star}^2]$.
We have $\pi=k^2\mathbb{E}[\hat{x}_n^2] = k^2(\mathbb{E}[x_n^2]-\sigma_\epsilon^2)$, which follows from $x_n=\hat{x}_n+\epsilon_n$ and the orthogonality property, i.e., $\mathbb{E}[\epsilon_n\hat{x}_n]=0$.
Given that $\bar{m}=\frac{\rho_{\rm p}}{1+\rho_{\rm p}}$, $\sigma_e^2=\frac{\pi\rho_{\rm p}}{(1+\rho_{\rm p})^2}$, and $k=I_1/\bar{m}$, we obtain $\sigma_e^2={I_1^2}/{\rho_{\rm p}}\left(\mathbb{E}[x_n^2]-\sigma_\epsilon^2\right)$.
To derive an analytical expression for $\mathbb{E}[x_n^2]$, we exploit the state-evolution equation, i.e., $x_{n+1}=(A-B\bar{m}_nk_n)\hat{x}_n+A\epsilon_n+Be_n+v_n$.
Taking the second-order moment of this equation gives
\begin{align}
    \mathbb{E}[x_{n+1}^{2}]
    &=(A-B\bar{m}_nk_n)^2\mathbb{E}[\hat{x}_{n}^{2}]
    +A^2\sigma_{\epsilon,n}^{2}
    +B^2\sigma_{e,n}^{2}
    +\sigma_v^{2}, \notag 
\end{align}
where the cross terms vanish due to the Kalman orthogonality
property and the statistical independence assumption.
When the control process reaches the steady state, we have
$I_2 (\rho_{\rm p}) \triangleq \mathbb{E}[x_n^2]=\mathbb{E}[x_{n+1}^2]$, i.e., the second moment of the plant state is time-invariant.
In addition, we also have $\bar{m}_nk_n\rightarrow I_1$,
$\sigma_{\epsilon,n}^2\rightarrow\sigma_\epsilon^2$, and
$\sigma_{e,n}^2\rightarrow\sigma_e^2$.
Therefore, we have
\begin{align}
     I_2 (\rho_{\rm p})=\frac{\left( A^2-(A-BI_1)^2 \right) \sigma _{\epsilon}^{2}+B^2\sigma _{e}^{2}+\sigma _{v}^{2}}{1-(A-BI_1)^2}. \label{eq:I2_stationary}
\end{align}
Define the normalized steady-state estimation-error
variance as $I_3\triangleq{\sigma_\epsilon^2}/{I_2 (\rho_{\rm p})}$.
Then, we can show that $I_3$ is independent of $\rho_{\rm p}$.
To this end, define the steady-state prior estimation-error variance as $\sigma_{\epsilon,-}^2 \triangleq A^2\sigma_\epsilon^2+B^2\sigma_e^2+\sigma_v^2$.
Stacking the real and imaginary parts of the received signal,
the corresponding posterior estimation-error variance is given by
\begin{align}
    \sigma _{\epsilon}^{2}=\frac{\sigma _{\epsilon ,-}^{2}}{1+\eta \sigma _{\epsilon ,-}^{2}/I_2 (\rho_{\rm p})}, \label{eq:scalar_kf_steady}
\end{align}
where $\eta \triangleq {2P_{\rm p}\|\mathbf{h}_{\rm f}\|_2^2}/{\sigma_{\rm f}^2}$.
Thus, the steady-state second-moment equation in \eqref{eq:I2_stationary} can be rewritten as follows:
\begin{align}
    I_2 (\rho_{\rm p}) = (A-BI_1)^2 \left(I_2 (\rho_{\rm p})-\sigma_\epsilon^2\right) +\sigma_{\epsilon,-}^2.
\end{align}
Dividing both sides by $I_2 (\rho_{\rm p})$ and using $I_3=\sigma_\epsilon^2/I_2 (\rho_{\rm p})$, we obtain
\begin{align}
    {\sigma_{\epsilon,-}^2}/{I_2}=1-(A-BI_1)^2(1-I_3). \label{eq:normalized_prior_error}
\end{align}
Hence, we have 
\begin{align}
    I_3=\frac{1-(A-BI_1)^2(1-I_3)}{1+\eta \left( 1-(A-BI_1)^2(1-I_3) \right)}. \label{eq:I3_fixed_point}
\end{align}
Since both $I_1$ and $\eta$ are independent of $\rho_{\rm p}$,
\eqref{eq:I3_fixed_point} is not parameterized by $\rho_{\rm p}$.

Next, we show how $I_2(\rho_{\rm p})$ depends on $\rho_{\rm p}$.
Using $\sigma_\epsilon^2=I_3I_2 (\rho_{\rm p})$, the actuation-disturbance variance is given by
\begin{align}
    \sigma_e^2=\frac{I_1^2}{\rho_{\rm p}} \left(I_2 (\rho_{\rm p})-\sigma_\epsilon^2\right) = \frac{I_1^2(1-I_3)I_2 (\rho_{\rm p})}{\rho_{\rm p}}. \label{eq:sigma_e_I3}
\end{align}
Substituting these two error variances into the steady-state
second-moment equation gives
\begin{align}
    I_2 (\rho_{\rm p})= I_4 I_2 (\rho_{\rm p})+A^2I_3I_2 (\rho_{\rm p})+\frac{B^2I_{1}^{2}(1-I_3)I_2 (\rho_{\rm p})}{\rho _{\mathrm{p}}}+\sigma _{v}^{2}, \notag
\end{align}
where $I_4 \triangleq 1-(A-BI_1)^2(1-I_3)-A^2I_3$ is independent of $\rho_{\rm p}$.
By algebraic manipulations, $I_2 (\rho_{\rm p})$ can be expressed as
\begin{align}
     I_2 (\rho_{\rm p}) = \frac{\sigma_v^2\rho_{\rm p}} {I_4\rho_{\rm p} -B^2I_1^2(1-I_3)}. \label{eq:I2_rho}
\end{align}
Recall that $I_2 (\rho_{\rm p})\triangleq\mathbb{E}[x_n^2]=\mathbb{E}[x_{n+1}^2]$ denotes the steady-state second moment of the plant state, indicating that $I_2 (\rho_{\rm p}) \ge 0$. 
Consequently, \eqref{eq:I2_rho} is positive and $I_4 > 0$.
Substituting \eqref{eq:I2_rho} into $\sigma_e^2$ and $\sigma_\epsilon^2$ gives
\begin{align}
    \sigma_e^2 &=\frac{I_1^2(1-I_3)\sigma_v^2}{I_4\rho_{\rm p}-B^2I_1^2(1-I_3)}, \notag \\
    \sigma_\epsilon^2 &=\frac{I_3\sigma_v^2\rho_{\rm p}}{I_4\rho_{\rm p} -B^2I_1^2(1-I_3)}. \notag
\end{align}
Substituting the above two expressions into \eqref{eq:cost_I1},
the infinite-horizon average control cost can be expressed as
\begin{align}
    J_\infty(\rho_{\rm p})=\theta\sigma_v^2+\frac{(D+B^2\theta)I_1^2\sigma_v^2(1-I_3+I_3\rho_{\rm p})}{I_4 \rho_{\rm p}-B^2I_1^2(1-I_3)}. \label{eq:cost_rho}
\end{align}
By taking the first-order derivative of \eqref{eq:cost_rho} w.r.t. $\rho_{\rm p}$, we have 
\begin{align}
   \frac{\mathrm{d}J_{\infty}}{\mathrm{d}\rho _{\mathrm{p}}}=-\frac{(D+B^2\theta )I_{1}^{2}\sigma _{v}^{2}(1-I_3)\left( I_4 +B^2I_{1}^{2}I_3 \right)}{(I_4 \rho _{\mathrm{p}}-B^2I_{1}^{2}(1-I_3))^2}. \label{eq:cost_derivative_rho}
\end{align}
Given that $I_3 = \sigma_\epsilon^2 / (\mathbb{E}[\hat{x}_n^2] + \sigma_\epsilon^2)$ and $\mathbb{E}[\hat{x}_n^2]>0$, we further have $I_3 \in [0,1)$.
Moreover, according to the discussion at the end of Appendix \ref{appendix:proof_infinite_horizon_scalar}, we have $\theta > 0$, since it is easy to prove from \eqref{eq:appendix_theta} that $\theta$ is positive under the assumption that $D>0$.
Therefore, whenever $I_1=\bar{m}k \neq 0$, \eqref{eq:cost_derivative_rho} is strictly negative, i.e., ${\mathrm{d}J_\infty}/{\mathrm{d}\rho_{\rm p}}<0$.
This completes the proof.

\bibliographystyle{IEEEtran}
\bibliography{reference}
\end{document}